\documentclass[runningheads,english,numberwithinsect, envcountsame]{llncs}
\usepackage[english]{babel}
\usepackage[utf8]{inputenc}

\usepackage{amsmath, amssymb}
\usepackage{cite}
\usepackage{enumerate}
\usepackage[mathlines]{lineno}
\usepackage{paralist} 
\usepackage{xspace} 
\usepackage{marcel_notation}
\usepackage{todonotes}
\usepackage{subfig}

\let\doendproof\endproof
\renewcommand{\endproof}{\hfill $\square$\doendproof}

\newcommand{\clusterDrawing}{\DiskMap-framed Drawings of Non-planar Arrangements}

\title{Drawing Clustered Graphs on Disk Arrangements\thanks{Work
was partially supported by grant WA 654/21-1 of the German Research Foundation
(DFG).}}

\author{Tamara Mchedlidze\inst{1} \and Marcel Radermacher\inst{1} \and Ignaz
Rutter\inst{2} \and Nina Zimbel\inst{1}}

\institute{
		Department of Computer Science, Karlsruhe  Institute of Technology,
	Germany \and
	Department of Computer Science and Mathematics, University
	of Passau, Germany \\
	\email{mched@iti.uka.de}, \email{radermacher@kit.edu}, \email{rutter@fim.uni-passau.de}}

\begin{document}

\maketitle

\begin{abstract}
	Let $G=(V, E)$ be a planar graph and let $\Clustering$ be a partition of $V$.
	We refer to the  graphs induced by the vertex sets in $\Clustering$ as
	\emph{clusters}.
	Let $\DiskMap$ be an arrangement of disks with a bijection between the disks
	and the clusters.  Akitaya et al.~\cite{DBLP:journals/corr/abs-1709-09209}
	give an algorithm to test whether $(G, \Clustering)$ can be embedded onto
	$\DiskMap$ with the additional constraint that edges are routed through a set
	of pipes between the disks.  Based on such an embedding, we prove that every
	clustered graph and every disk arrangement without pipe-disk intersections has
	a planar straight-line drawing where every vertex is embedded in the disk
	corresponding to its cluster. This result can be seen as an extension of the
	result by Alam et al.~\cite{JGAA-367} who solely consider biconnected
	clusters. Moreover, we prove that it is $\cNP$-hard to decide whether a
	clustered graph has such a straight-line drawing, if we permit pipe-disk
	intersections.
%
\end{abstract}

\section{Introduction}
\begin{figure}[t]
	\centering
	\subfloat[Conditions (C1) and (C2)]{
		\includegraphics[page=1]{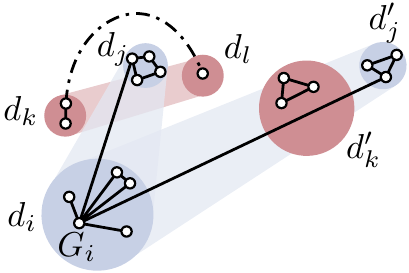}
	}
	\quad
	\subfloat[Conditions (P1) and (P2)]{
		\includegraphics[page=3]{fig/disk_arrangement.pdf}
	}

	\caption{(a) The blue disk arrangement satisfies the conditions (C1,
		C2) and (P1, P2). The disks $d_k, d_l$ and $d_i, d_j$ violate condition
		(C1). The disks $d'_k$ and $d_i, d'_j$ violate (C2). Note that the edge from
		disk $d_i$ to $d'_j$ has to cross the boundary of $d'_k$ twice. (b) The disks
		$d_k, d_i$ violate condition (P1) and $d'_k$ and $d_i, d_j$ violate
		condition (P2).
	}
	\label{fig:disk_arrangement}
\end{figure}

In practical applications, it often happens that a graph drawing produced by an
algorithm has to be post processed by hand to comply with some particular
requirements. Thus, the user moves vertices and modifies edges in order to fulfill
these requirements.  Interacting with large graphs is often time-consuming. It
takes a lot of time to group and move the vertices or process them individually
and to control the overall appearance of the produced drawing.  The problem we
study in this paper addresses this scenario. In particular, we assume that a
user wants to modify a drawing of a large planar graph $G$. Instead we provide
her with an abstraction of this graph.  The user modifies the abstraction and
thus providing some constraints on how the drawing of the initial graph should
look like.  Then our algorithm propagates the drawing of the abstraction to the
initial graph so that the provided constraints are satisfied.

More formally, we model this scenario in terms of a \emph{(flat) clustering} of
a graph $G=(V,E)$, i.e., a partition $\Clustering  = \{V_1, \dots, V_k\}$ of the
vertex set $V$.  We refer to the pair $\Cluster = (G, \Clustering)$ as a
\emph{clustered graph} and the graphs $G_i$ induced by $V_i$ as \emph{clusters}.
The set of edges $E_i$ of a cluster $G_i$ are \emph{intra-cluster edges} and the
set of edges with endpoints in different clusters \emph{inter-cluster edges}.
A \emph{disk arrangement} $\DiskArr =\{\Disk_1, \dots, \Disk_k \}$ is a set of
pairwise disjoint disks in the plane together with a bijective mapping $\mu(V_i)
= \Disk_i$ between the clusters $\Clustering$ and the disks $\DiskArr$.  We refer
to a disk arrangement $\DiskArr$ with a bijective mapping $\mu$ as a \emph{disk
arrangement of $\Cluster$}, denoted by $\DiskMap$. 
A \emph{$\DiskMap$-framed drawing of a clustered graph $\Cluster=(G,
\Clustering)$} is a planar drawing of $G$ where each cluster $G_i$ is drawn
within its corresponding disk $d_i$.  We study the following problem: given a
clustered planar graph $\Cluster = (G, \Clustering)$, an embedding $\embedding$
of $G$ and a disk arrangement $\DiskMap$ of $\Cluster$, does $\Cluster$ admit a
$\DiskMap$-framed straight-line drawing homeomorphic to $\embedding$?


\myparagraph{Related Work}
Feng et al.~\cite{DBLP:conf/esa/FengCE95} introduced the notion of \emph{clustered
graphs} and \emph{c-planarity}. A graph $G$ together with a recursive partitioning
of the vertex set is considered to be a clustered graph. An embedding of $G$
is \emph{c-planar} if
\begin{inparaenum}[(i)]
\item each cluster $c$ is drawn within a connected region $R_c$, 
\item two regions $R_c, R_d$ intersect if and only if the cluster $c$ contains
	the cluster $d$ or vice versa, and 
\item every edge intersects the boundary of a region at most once.
\end{inparaenum}
They prove that a c-planar embedding of a connected clustered graph can be
computed in $O(n^2)$ time.  It is an open question whether this result can be
extended to disconnected clustered graphs. Many special cases of this
problem have been considered~\cite{BLASIUS2016306}.

Concerning drawings of c-planar clustered graphs, Eades et al.~\cite{Eades2006}
prove that every c-planar graph has a c-planar straight-line drawing where each
cluster is drawn in a convex region.  Angelini et
al.~\cite{DBLP:journals/dcg/AngeliniFK11} strengthen this result by showing that
every c-planar graph has a c-planar straight-line drawing in which every cluster
is drawn in an axis-parallel rectangle. The result of Akitaya et
al.~\cite{DBLP:journals/corr/abs-1709-09209} implies that in $O(n \log n)$ time
one can decide whether an abstract graph with a flat clustering has an embedding
where each vertex lies in a prescribed topological disk and every edge is routed
through a prescribed topological pipe.  In general they ask whether a simplicial
map $\varphi$ of $G$ onto a 2-manifold $M$ is a \emph{weak embedding}, i.e., for
every $\epsilon > 0$, $\varphi$ can be perturbed into an embedding
$\psi_\epsilon$ with $||\varphi - \psi_\epsilon|| < \epsilon$.


Godau~\cite{10.1007/3-540-58950-3_377} showed that  it is $\cNP$-hard to decide
whether an embedded graph has a $\DiskMap$-framed straight-line drawing. The
proof relies on a disk arrangement $\DiskMap$ of overlapping disks that have
either radius zero or a large radius. 

Banyassady et
al.~\cite{10.1007/978-3-319-62127-2_7} study whether the intersection graph of
unit disks has a straight-line drawing such that each vertex lies in its disk.
They proved that this problem is $\cNP$-hard regardless of whether the embedding
of the intersection graph is prescribed or not. Angelini et
al.~\cite{10.1007/978-3-662-45803-7_34} showed it is $\cNP$-hard to decide
whether an abstract graph $G$ and an arrangement of unit disks have a
$\DiskMap$-framed straight-line drawing. They leave the problem of finding a
$\DiskMap$-framed straight-line drawing of $G$ with a fixed embedding as an open
question.  Alam et al.~\cite{JGAA-367} prove that it is \cNP-hard to decide
whether an embedded clustered graph has a c-planar straight-line drawing where
every cluster is contained in a prescribed (thin) rectangle and edges have to
pass through a defined part of the boundary of the rectangle. Further, they
prove that all instances with biconnected clusters always admit a solution.
Their result implies that graphs of this class have $\DiskMap$-framed
straight-line drawings.

Rib\'o~\cite{ribo2006realization} shows that every embedded clustered graph
where each cluster is a set of independent vertices has a straight-line drawing
such that every cluster lies in a prescribed disk. In contrast to our setting
Rib\'o allows an edge $e$ to intersect a disk of a cluster $G_i$ that
does not contain an endpoint of $e$.

\myparagraph{Contribution} 
A \emph{pipe $p_{ij}$ of two clusters $V_i,
V_j$} is the \emph{convex hull} of the disks $\Disk_i$ and $\Disk_j$, i.e., the
smallest convex set of points containing $\Disk_i$ and $\Disk_j$; see
Fig.~\ref{fig:disk_arrangement}. 
We refer to a topological planar drawing of $G$ as an \emph{embedding of $G$}.
A \emph{$\DiskMap$-framed embedding of $G$} is a $\DiskMap$-framed topological
drawing of $G$ with the additional requirement that \begin{inparaenum}[(i)]
\item each intra-cluster edge entirely lies in its disk \item each inter-cluster
edge $uv$ intersects with a pipe $p_{ij}$ if and only if $u$ and $v$ are
vertices of the clusters $G_i$ and $G_j$, respectively, and \item each edge
	crosses the boundary of a disk at most once.  \end{inparaenum} This concept is
also known as \emph{c-planarity with embedded pipes}~\cite{CORTESE20091856}.
An embedding $\embedding$ of $G$ is \emph{compatible with $\DiskMap$} if
$\embedding$ is homeomorphic to a $\DiskMap$-framed embedding of $G$. The result
of Akitaya et al. can be used to decide whether an embedding $\embedding$ of
$G$ is compatible with $\DiskMap$.

The following two conditions are necessary, for $\Cluster$ to have a
$\DiskMap$-framed embedding:
\begin{inparaenum}[(C1)]
	\item if $(V_i \times V_j) \cap E \not= \emptyset$ and $(V_k \times V_l) \cap
		E \not= \emptyset$ ($i,j, k, l$ pairwise distinct), then the intersection of
		the pipes $p_{ij}$ and $p_{kl}$ is empty, and
\item the set $p_{ij} \setminus \Disk_k$ is connected.
\end{inparaenum}
Thus, in the following we assume that $\DiskMap$ satisfies (C1) and (C2).  A
\emph{planar} disk arrangement additionally satisfies the condition that
\begin{inparaenum}[(P1)] \item the pairwise intersections of all disks are
		empty, and
	\item $(V_i \times V_j) \cap E \not=\emptyset$, the intersection of $p_{ij}$
		with all disks $d_k$ (corresponding to $V_k$) is empty ($i,j,k$ pairwise
		distinct). 
\end{inparaenum}
A planar disk arrangement can be seen as a thickening of a planar straight-line
drawing of the graph obtained by contracting all clusters.

We prove that every clustered graph $(G, \Clustering)$ with planar disk
arrangement $\DiskMap$ and an $\DiskMap$-framed embedding $\embedding$ has a
$\DiskMap$-framed planar straight-line drawing homeomorphic to~$\embedding$.
Taking the result of Akitaya et al.~\cite{DBLP:journals/corr/abs-1709-09209}
into account, our result can be used to test whether an abstract clustered graph
with connected clusters has a $\DiskMap$-framed straight-line drawing. Cluster
$G_i$ in Fig.~\ref{fig:disk_arrangement} shows that in general clusters cannot
be augmented to be biconnected, if the embedding is fixed. Hence, our result is
generalization of the result of Alam et al.~\cite{JGAA-367}.  In
Section~\ref{sec:hardness} we show that the problem is $\cNP$-hard in the case
that the disk arrangements does not satisfy condition (P2).  From now on we
refer to a planar straight-line drawing of $G$ simply as a drawing of $G$.

\section{Drawing on Planar Disk Arrangements}

\newcommand{\Cin}{\Cluster_{\mathrm{in}}}
\newcommand{\Cout}{\Cluster_{\mathrm{out}}}

\begin{figure}[t]
	\centering
	\includegraphics[page=2,width=0.4\textwidth]{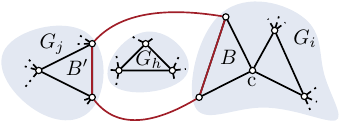}
	\caption{A planar clustered graph $\Cluster$ that is not simple.
	}
	\label{fig:non_simple_clustered_graph}
\end{figure}

In this Section we prove that every \emph{simple} clustered graph with a planar
disk arrangement $\DiskMap$ and $\DiskMap$-framed embedding has a
$\DiskMap$-framed drawing.  
An embedded clustered
graph $\Cluster$ is \emph{simple} if for every $i,j$, there is no cluster $G_h
(i,j \neq h)$ embedded in the interior of the subgraph induced by $V_i \cup
V_j$; see Fig.~\ref{fig:non_simple_clustered_graph}.
Note that this is a necessary condition for the corresponding disk arrangement
to be planar. A clustered graph $\Cluster  = (G, \Clustering)$ is
\emph{connected} if each cluster $G_i$ is connected.  

We prove the statement by induction on the number of intra-cluster edges. In
Lemma~\ref{lem:contraction} we show that we can indeed reduce the number of
intra-cluster edges by contracting intra-cluster edges. In
Lemma~\ref{lemma:draw_with_outerface}, we prove that the statement is correct if
the outer face is a triangle and $\Cluster$ is connected.  In
Theorem~\ref{thm:draw_planar} we extend this result to clustered graphs whose
clusters are not connected.

Let $\Cluster = (G, \Clustering)$ with a disk arrangement $\DiskMap$ and a
$\DiskMap$-framed embedding $\embedding$. Let $uv$ be an intra-cluster edge of
$G$ that is not an edge of a separating triangle. We obtain a \emph{contracted
clustered graph $\Cluster / e$} of $\Cluster$ be removing $v$ from $G$ and
connecting the neighbors of $v$ to $u$. We obtain a corresponding embedding
$\embedding / e$ from $\embedding$ by routing the edges $vw \in E, w\not=u$
close to $uv$.

\begin{lemma} \label{lem:contraction}
	Let $\Cluster = (G, \Clustering)$ be a connected simple clustered graph with a
	planar disk arrangement $\DiskMap$ and a $\DiskMap$-framed embedding
	$\embedding$. Let $e$ be an intra-cluster edge that is not an edge of a
	separating triangle. Then $\Cluster$ has a $\DiskMap$-framed drawing that is
	homeomorphic to $\embedding$ if $\Cluster / e$ has a $\DiskMap$-framed drawing
	that is homeomorphic to $\embedding/e$.
\end{lemma}

\begin{proof} 
	Let $e=uv$ and denote by $u_0, u_1, \dots, u_k$ the neighbors of $u$ and $v_0,
	v_1, \dots, v_l$ the neighbors of $v$ in $\Cluster$. Without loss of
	generality, we assume that $u_0 = v$ and $v_0 = u$. Since $e$ is not an edge
	of a separating triangle the set $I := \{u_2, \dots, u_{k-1}\} \cap \{v_2,
	\dots, v_{l-1}\}$ is empty.
	Denote by $u$ the vertex obtained by the contraction of $e$. 
	Let $G_i$ be the cluster of $u$ and $v$,  and let $\Disk_i$ be the
	corresponding disk in $\DiskMap$. 

	Consider a $\DiskMap$-framed drawing $\Gamma/e$ of $\Cluster/e$ homeomorphic
	to $\embedding / e$.  Then there is a small disk $d_u \subset \Disk_i$ around
	$u$ such that for every point $p$ in $d_u$ moving $u$ to $p$ yields a
	$\DiskMap$-framed drawing that is homeomorphic to~$\embedding/e$.

	We obtain a straight-line drawing $\Gamma$ of $\Cluster$ from $\Gamma/e$ as
	follows. First, we remove the edges $uv_i$ from $\Gamma / e$.  The edges $u_1,
	u_k$ partitions $d_u$ into two regions $r_u, r_v$ such that the intersection
	of $r_v$ with $uu_i$ is empty for all $i \in \{2, \dots, k-1\}$. We place $v$
	in $r_v$ and connect it to $u$ and the vertices $v_1, \dots, v_l$. 
	Since $r_v$ is a subset of $d_u$ and $I = \emptyset$, we have that the new
	drawing $\Gamma$ is planar. Since $v$ is placed in $r_v$, the edge $uv$ is in
	between $u_1$ and $u_k$ in the rotational order of edges around $u$. Hence,
	$\Gamma$ is homeomorphic to $\embedding$.  Finally, $\Gamma$ is a
	$\DiskMap$-framed drawing since, $d_u$ is entirely contained in $\Disk_i$ and
	thus are $u$ and $v$.
\end{proof}

\begin{lemma} \label{lemma:draw_with_outerface}
	Let $\Cluster$ be a connected simple clustered graph with a triangular outer
	face $T$, a planar disk arrangement $\DiskMap$, and a $\DiskMap$-framed
	embedding $\embedding$. Moreover, let $\Gamma_T$ be a $\DiskMap$-framed
	drawing of $T$. Then $\Cluster$ has a $\DiskMap$-framed drawing that is
	homeomorphic to $\embedding$ with the outer face drawn as $\Gamma_T$.
\end{lemma}

\begin{proof}
	We prove the theorem by induction on the number of intra-cluster edges.
	
	First, assume that every intra-cluster edge of $\Cluster$ is an edge on the
	boundary of the outer face. Let $\Gamma$ be the drawing obtained by placing
	every interior vertex on the center point of its corresponding disk and draw
	the outer face as prescribed by $\Gamma_T$. Since $\DiskMap$ is a
	planar disk arrangement and $\Gamma_T$ is convex, the resulting drawing is
	planar and thus a $\DiskMap$-framed drawing of $\Cluster$ that is homeomorphic
	to the embedding $\embedding$.

	Let $S$ be a separating triangle of $\Cluster$ that splits $\Cluster$ into two
	subgraphs $\Cin$ and $\Cout$ so that $\Cin \cap \Cout = S$ and the outer face
	$\Cout$	and $\Cluster$ coincide. Then by the induction hypothesis $\Cout$ has
	the $\DiskMap$-framed drawing $\Dout$ with the outer face drawn as $\Gamma_T$
	and $\Cin$ as a $\DiskMap$-framed drawing $\Din$ with the outer face drawing
	as $\Dout[S]$, where $\Dout[S]$ is the drawing of $S$ in $\Dout$.  Then we
	obtain a $\DiskMap$-framed drawing of $\Cluster$ by merging $\Din$ and
	$\Dout$.  

	Consider an intra-cluster edge $e$ that does not lie on the boundary of the
	outer face and is not an edge of a separating triangle. Then by the induction
	hypothesis, $\Cluster/e$ has a $\DiskMap$-framed drawing with the outer face
	drawn as $\Gamma_T$. It follows by Lemma~\ref{lem:contraction} that $\Cluster$ has a
	$\DiskMap$-framed drawing homeomorphic to $\embedding$.
\end{proof}

\begin{theorem}
	\label{thm:draw_planar}
	Every simple clustered graph $\Cluster$ with a $\DiskMap$-framed embedding
	$\embedding$ has a $\DiskMap$-framed drawing homeomorphic to $\embedding$.
\end{theorem}

\begin{proof}
 We obtain a clustered graph $\Cluster'$ from $\Cluster$ by adding a new
 triangle $T$ to the graph and assigning each vertex of $T$ to is own cluster.
 Let $\Gamma_T$ be a drawing of $T$ that contains all disks in $\DiskMap$
 in its interior. We obtain a new disk arrangement $\DiskMap'$ from $\DiskMap$ by
 adding a sufficiently small disk for each vertex of $\Gamma_T$. The embedding
 $\embedding$ together with $\Gamma_T$ is a $\DiskMap'$-framed embedding
 $\embedding'$ of $\Cluster'$.

 According to Feng. et al.~\cite{DBLP:conf/esa/FengCE95} there is a simple
 connected clustered graph $\Cluster''$ that contains $\Cluster'$ as a subgraph
 whose embedding $\embedding''$ is $\DiskMap$-framed and contains $\embedding'$.
 By Lemma~\ref{lemma:draw_with_outerface} there is a
 $\DiskMap$-framed drawing $\Gamma''$ of $\Cluster''$ homeomorphic to
 $\embedding''$ with the outer face drawn as $\Gamma_T$. The drawing $\Gamma''$
 contains a $\DiskMap$-framed drawing of $\Cluster$.
\end{proof}

\section{Drawing on General Disk Arrangements} \label{sec:hardness} We study the
following problem referred to as \textsc{\clusterDrawing}.
Given a planar clustered graph $\Cluster=(G, \mathcal V)$, a disk arrangement
$\DiskMap$ that is not planar, i.e., $\DiskMap$ satisfies condition (C1) and
(C2) but not (P1) and (P2), and a $\DiskMap$-framed embedding $\embedding$ of
$G$, is there a $\DiskMap$-framed straight-line drawing $\Gamma$ that is
homeomorphic to $\embedding$ and $\DiskMap$?
%
%
Note that if the disks $\DiskMap$ are allowed to overlap (condition (P1)) and
$G$ is the intersection graph of $\DiskMap$, the problem is known to be
$\cNP$-hard~\cite{10.1007/978-3-319-62127-2_7}. Thus, in the following we
require that the disks do not overlap, but there can be disk-pipe intersections,
i.e, $\DiskMap$ satisfies conditions (C1), (C1) and (P1) but not (P2).
By Alam at al.~\cite{JGAA-367} it follows that the problem restricted to thin
touching rectangles instead of disks is $\cNP$-hard.  We strengthen this result
and prove that in case that the rectangles are axis-aligned squares and are not
allowed to touch the problem remains $\cNP$-hard.  Our illustrations contain
blue dotted circles that indicate how the square in the proof can be replaced by
disks.  For the entire proof we refer to Appendix~\ref{apx:hardness}.

To prove $\cNP$-hardness we reduce from \textsc{Planar Monotone
3-SAT}~\cite{DBLP:journals/ijcga/BergK12}. For each literal and clause we
construct a clustered graph $\Cluster$ with an arrangement of squares $\DiskMap$ of
$\Cluster$ such that each disk contains exactly one vertex.  We refer to these
instances as \emph{literal} and \emph{clause gadgets}.  In order to transport
information from the literals to the clauses, we construct a \emph{copy} and
\emph{inverter gadget}.  The design of the gadgets is inspired by Alam et
al.~\cite{JGAA-367}, but due to the restriction to squares rather than
rectangles, requires a more careful placement of the geometric objects.  The
green  and red regions in the figures of the gadget correspond to
\emph{positive} and \emph{negative} drawings of the literal gadget.  The green
and red line segments indicate that for each truth assignment of the
variables our gadgets indeed have $\DiskMap$-framed straight-line drawings. Negative versions
of the literal and clause gadget are obtained by mirroring vertically. Hence,
we assume that variables and clauses are positive.  Each gadget covers a set of
checkerboard cells. This simplifies the assembly of the gadgets for the
reduction.

An \emph{obstacle of a pipe $p_{ij}$} is a disk $d_k, i,j \not=k,$ that
intersects $p_{ij}$.  The \emph{obstacle number of a pipe $p_{ij}$} is the
number of obstacles of $p_{ij}$. Let $P = \{p_{ij} \mid V_i \times V_j \cap E
\not= \emptyset\}$.  The \emph{obstacle number of a disk arrangement $\DiskMap$}
is maximum obstacle number of all pipes $p_{ij}$ with $V_i \times V_j \cap E
\not= \emptyset$.

\begin{figure}[tb] \centering \includegraphics[page=3]{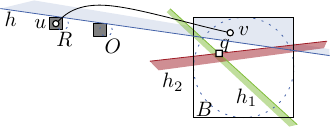}
	\caption{Regulator} \label{fig:sec:regulator} \end{figure}

\myparagraph{Regulator}
The \emph{regulator gadget} restricts the feasible placements of a  vertex $v$
that lies in the interior of a square $B$; refer to
Fig.~\ref{fig:sec:regulator}. Let $h_1, h_2$ be two half planes such
that the intersection $q$ of their supporting lines lies in $B$. 
In a $\DiskMap$-framed drawing of the regulator gadget the placement of $v$ is
restricted by a half plane $h$ that excludes a placement of $v$ in $h_1 \cap
h_2$ but allows for a placement in $h_1 \cap B$ or $h_2 \cap B$. We refer to  $h
\cap B$ as the \emph{regulated region of $B$}.

\myparagraph{Literal Gadget}
The \emph{positive literal gadget} is depicted in
Fig.~\ref{fig:sec:hardness:literal}.  The \emph{center block} is a unit square $C$
with corners $\alpha_1, \alpha_2, \alpha_3, \alpha_4$ in clockwise order.  For
each corner $\alpha_i$ of $C$ consider a line $l_i$ that is tangent to $C$ in
$\alpha_i$, i.e, $l_i \cap C = \{\alpha_i\}$. Let $p_i$ be the intersection of
lines $l_{i-1}$ and $l_{i}$ where $l_0 = l_4$; refer to
Fig.~\ref{fig:sec:hardness:literal:constr}.  Let $R_1, \dots, R_4$ be four pairwise
non-intersecting squares that are disjoint from $C$ such that $R_i$ contains
$p_i$ in its interior.  We add a cycle $v_1v_2v_3v_4v_1$ such that $v_i \in
R_i$. We refer to the vertex $v_i$ as the \emph{cycle vertex} of the \emph{cycle
block $R_i$}. For each $i$, let $j_i$ be a half plane that contains $R_{i+1}$
but does not intersect $C$.  We place a regulator $W_i$ of $v_i$  with respect
to $h_{i-1}$ and $h_{i}$ and position it such that it lies in $j_i$, where $h_i$
is the half plane spanned by $l_i$ with $C \not\subseteq h_i$. 

\begin{figure}[t]
	\centering
	\subfloat[]{
		\includegraphics[page=4]{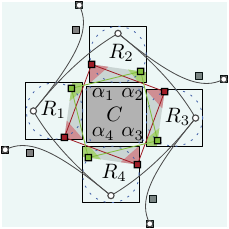}
	}
	\quad\quad
	\subfloat[\label{fig:sec:hardness:literal:constr}]{
		\includegraphics[page=5]{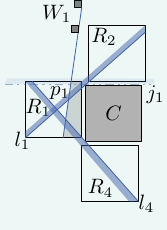}
	}
	\caption{Literal gadget}
	\label{fig:sec:hardness:literal}
\end{figure}

We now describe the two combinatorially different realizations of the literal
gadgets. 
Consider $R_1$ and its two adjacent squares $R_2$ and $R_4$.  Let $Q_i$ be the
regulated region of $R_i$ with respect to $W_i$.  We refer to  $\overline{h_2}
\cap \overline{h_4} \cap Q_1$ as the \emph{infeasible
region of $R_1$},
where $\overline{h_i}$ denotes the complement of $h_i$.  The intersection $h_1
\cap Q_1$ is the \emph{positive region $P_1$ of $R_1$}.  The region
$\overline{h_4} \cap Q_1$ is the \emph{negative region $N_1$ of $R_1$}.  All
these regions are by construction not empty. 
The positive, negative and infeasible region of $R_i, i \not= 1$ are defined
analogously.

\newcommand{\PropertyNoInfeasibleDrawing}{
	If $\Gamma$ is a $\DiskMap$-framed drawing of a positive (negative) literal
	gadget, then no cycle vertex $v_i$ lies in the infeasible region of $R_i$.
	Moreover, either each cycle vertex $v_i$ lies in the positive region~$P_i$ or
	each vertex $v_i$ lies in the negative region~$N_i$. 
}
\begin{property} \label{lemma:no_infeasible_drawing}
	\PropertyNoInfeasibleDrawing
\end{property}

\newcommand{\PropertyValidLiteral}{
	The positive and negative placements induce a $\DiskMap$-framed drawing of the
	literal gadget, respectively.
}
\begin{property} \label{lemma:valid_literal_gadget}
	\PropertyValidLiteral
\end{property}

\myparagraph{Copy and Inverter Gadget} The copy gadget in
Fig.~\ref{fig:sec:hardness:copy} connects two positive literal gadgets $X$ and $Y$
such that a drawing of $X$ is positive if and only if the drawing of $Y$ is
positive. The inverter gadget connects a positive literal gadget $X$ to a
negative literal gadget $Y$ such that the drawing of $X$ is positive if and only
if the drawing of $Y$ is negative. The construction of the gadgets uses ideas
similar to the construction of the literal gadget. In contrast to the literal
gadget, we replace the center block by four squares. 

\begin{figure}[tb]
	\centering
	\includegraphics[page=4]{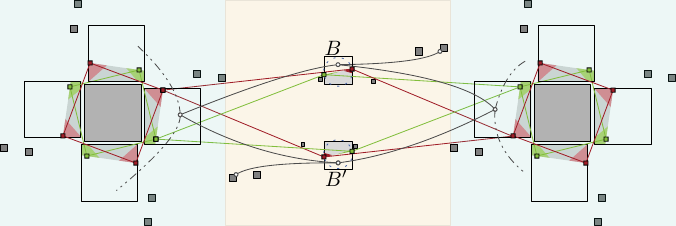}
	\caption{Copy gadget}
	\label{fig:sec:hardness:copy}
\end{figure}

\newcommand{\PropertyConsistentCopyGadget}{
	Let $\Gamma$ be a $\DiskMap$-framed drawing of two positive (negative) literals
	gadgets $X$ and $Y$ connected by a copy gadget. Then the $\DiskMap$-framed
	of $X$ in $\Gamma$ is positive if and only if the $\DiskMap$-framed
	drawing of $Y$ is positive.
}

\begin{property}
	\label{lemma:consistent_copy_gadget}
	\PropertyConsistentCopyGadget
\end{property}

\newcommand{\PropertyValidCopyInverter}{
	The positive (negative) placement of two literals gadgets $X, Y$ induces a
	$\DiskMap$-framed drawing of a copy [inverter] gadget  that connects $X$ and
	$Y$.
}

\begin{property} \label{lemma:valid_copy_inverter_gadget}
	\PropertyValidCopyInverter
\end{property}

\iffalse
\myparagraph{Clause Gadget}
We construct a \emph{clause gadget} with respect to three positive literal
gadgets $X, Y,Z$ as arranged as in Fig.~\ref{fig:sec:hardness:clause}.  The
clause gadget transports the information of each literal $A$ with a
\emph{transition block $T_A$} to the vertex $k$ in the interior of the
\emph{clause block $K$}.  The blue half-plane $h_A, A\in\{A, B\}$ in the figure
indicates that if $k \in \overline{h_A}$ the drawing of literal gadget $A$
cannot be negative.  Since $h_X \cap h_Y \cap  h_Z = \emptyset$, there is no
$\DiskMap$-framed drawing such that all literals have a negative drawing at the
same time.  For each $A, B \in \{X, Y, Z\}, A\not=B$, let $q^+_A, q^-_A$ and
$q_{A, B}$ be the points as depicted in the Fig.~\ref{fig:sec:hardness:clause}.
The gadget is constructed so that for each $y \in \{q^-_Y, q^+_Y\}$ and $z \in
\{q^-_Z, q^+_Z\}$ the points $y, z$ and $q_{Y, Z}, q^{+X}$ induce a
$\DiskMap$-framed drawing of the clause gadget. The analog statement for $q_{X,
Z}, q_{X, Y}$ is also true.

\begin{figure}[tb]
	\centering
	\includegraphics[page=6]{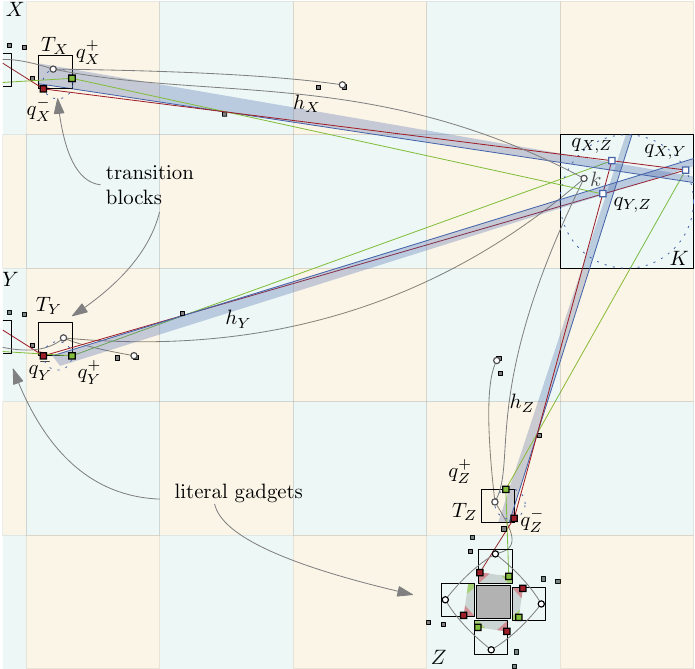}	
	\caption{Construction of the clause gadget.}
	\label{fig:sec:hardness:clause}
\end{figure}

\newcommand{\PropertyValidClause}{
	There is no $\DiskMap$-framed drawing of the clause gadget such that the
	drawing of each literal gadget is negative. For all other 
	combinations of
	positive and negative drawings of the literal gadgets 
	there is
	a $\DiskMap$-framed drawing of the clause gadget.
}

\begin{property} \label{lemma::valid_clause_gadget}
	\PropertyValidClause
\end{property}

\else

\label{sec:clause}


\myparagraph{Clause Gadget}
We construct a \emph{clause gadget} with respect to three positive literal
gadgets $X, Y,Z$ arranged as depicted in Fig.~\ref{fig:sec:clause_block}. The
negative clause gadget, i.e., a clause with three negative literal gadgets, is
obtained by mirroring vertically.

We construct the clause gadget in two steps. First, we place a \emph{transition
block} $T_A$ close to each literal gadget $A \in \{X, Y, Z\}$. In the second
step, we connect the transition block to a vertex $k$ in a \emph{clause block
$K$} such that for every placement of $k$ in $K$ at least one drawing of the
literal gadgets has to be positive.

\begin{figure}[tb]
	\centering
	\includegraphics[page=2]{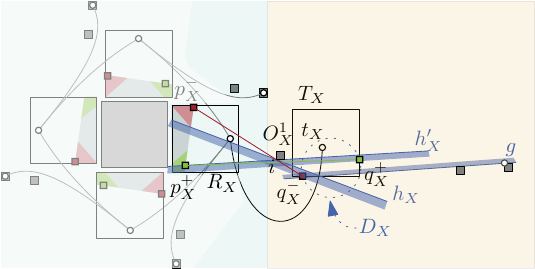}	
	\caption{Construction of the transition block.}
	\label{fig:sec:transition_block}
\end{figure}

Consider the literal gadget $X$ and let $R_X$ be the right-most cycle block of
$X$.  Let $h_X$ be a negative half plane of $R_X$, i.e., $h_X$ contains the
positive region $P_X$ but not the negative region $N_X$, refer to
Fig.~\ref{fig:sec:transition_block}. We now place a transition block $T_X$ such
that the intersection $T_X \cap h_X$ has small area.  Further, let $p^+_X$ and
$p^-_X$ be the positive and negative placements of $X$, respectively. Let
$q^-_X$ be a point in $T_X \cap h_X$. Let $i$ be the intersection point of the
supporting line $l_X$ of $h_X$ and the line segment $p^-_Xq^-_X$. We place an
obstacle $O^1_X$ such that $l_X$ is tangent to $O^1_X$ in point~$i$.  Finally,
we place a \emph{transition vertex} $t_X$ in the interior of $T_X$ and route the
edge $v_Xt_X$ through $h_X \cup T_X \cup R_X$, where $v_X \in R_X$.

Consider a half plane $h'_X$  such that $O^1_X \not\subseteq h'_X$ and $N_X
\not\subseteq h'_X$ and such that the supporting line $l'_X$ of $h'_X$ contains
$p^+_X$ and is tangent to $O^1_X$. Let $q^+_X$ be a point $h'_X \cap R_X$.
Observe that for $q^+_X$ and $q^-_X$ there is a positive and negative drawing of $X$,
respectively.  Further, if $X$ has a negative drawing, then $t_X$ lies in the
region $h_X \cap T_X$. In the following, we refer to $h_X \cap T_X$ as the
\emph{negative region} of $T_X$. The transition blocks of $Y$ and $Z$ are
constructed analogously with only minor changes. The transition block $T_Z$ of
$Z$ is constructed with respect to the top-most cycle block. Note that we can
choose the points $q^+_A, A \in \{X, Y, Z\}$ independent from each other as long
as each of them induces a positive drawing the literal gadget $A$.

Denote by $\xmax S$ the maximum $x$-coordinate of a point in a bounded set $S
\subset \R^2$.  Note that $\xmax D_X\cap h_X > \xmax T_X \cap h_X$, refer to
Fig.~\ref{fig:sec:transition_block}. To ensure that our construction remains
correct for disks we add a regulator $R$ with a respect a half plane $g$ such
that $\xmax D_X\cap h_x \cap g = \xmax T_X \cap h_x \cap g$ and $g$ contains
$q^+_X, q^-_X$.

\begin{figure}[tb]
	\centering
	\includegraphics[page=6]{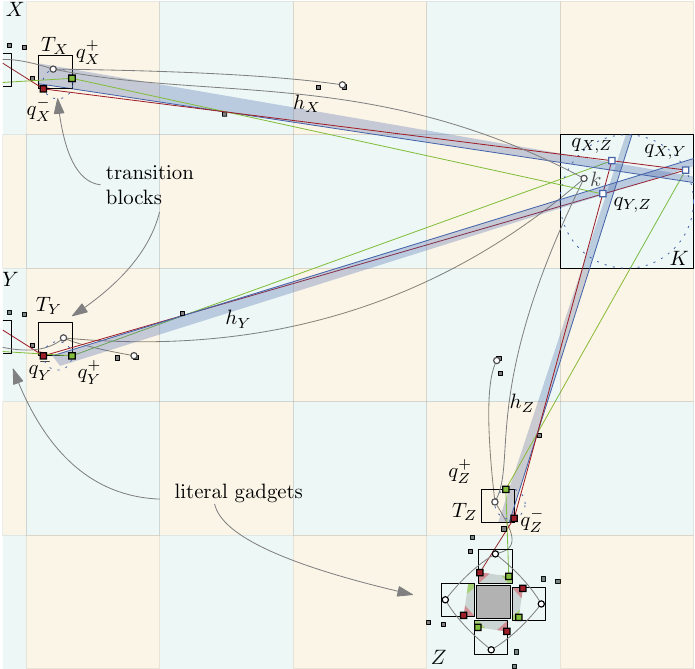}	
	\caption{Construction of the clause block.}
	\label{fig:sec:clause_block}
\end{figure}

Given the placement of the transition block $T_X, T_Y$ and $T_Z$ as depicted in
Fig.~\ref{fig:sec:clause_block}, we construct the \emph{clause block $K$} as
follows. We choose a point $q_{X,Y}$. Let $l^-_X$ and $l^-_Y$ be the lines
through the points $q^-_X, q_{X, Y}$, and $q^-_Y, q_{X, Y}$, respectively.
Further, consider a line $l^-_Z$ with $q^-_Z  \in l^-_Z$ such that the
intersection point $q_{A, Z} := l^-_Z \cap l^-_A, A \in \{X, Y\}$ lies in
between $q^-_A$ and $q_{X, Y}$. Further, let $l^+_X$ be the line through $q^+_X,
q_{Y, Z}$, $l^+_Y$ the line through $q^+_Y, q_{X, Z}$, and let $l^+_Z$ be the line through
$q^+_Z$ and $q_{X, Y}$. Let $h_A$ be a half plane that does not contain the
negative region $N_A$ and whose supporting line contains the intersection $i_A$
of $l^-_A$ and $l^+_A$. We place obstacles $O^2_A$ such that $O^2_A
\not\subseteq h_A$ and the supporting line of $h_A$ is tangent to $O^2_A$ in
point $i_A$. We place the clause box $K$ such that it contains $q_{X, Y},
q_{Y,Z}$, $q_{X,Z}$ and a new vertex $k$ in its interior. We finish the
construction by routing the edges $kt_A$ through $K \cup h_A \cup T_A, A \in
\{X, Y, Z\}$, where $t_A \in T_A$.

By construction we have that for each $y \in \{q^-_Y, q^+_Y\}$ and $z \in
\{q^-_Z, q^+_Z\}$ the points $y, z$ and $q_{Y, Z}$ induce a $\DiskMap$-framed
drawing. The analog statement for the points $q_{X, Z}$ and $q_{X, Y}$ is also
true.  Further, if $h_X \cap h_Y \cap  h_Z = \emptyset$, then there is no
$\DiskMap$-framed drawing such that each vertex $t_A$ lies on $q^-_A$.
Fig.~\ref{fig:sec:clause_block} shows that there is an arrangement of the clause
block and the obstacles such that $h_X \cap h_Y \cap  h_Z$ indeed is empty. 

\newcommand{\PropertyValidClause}{
	There is no $\DiskMap$-framed drawing of the clause gadget such that the
	drawing of each literal gadget is negative. For all other 
	combinations of
	positive and negative drawings of the literal gadgets 
	there is
	a $\DiskMap$-framed drawing of the clause gadget.
}
\begin{property} \label{lemma::valid_clause_gadget}
	\PropertyValidClause
\end{property}

\fi

\myparagraph{Reduction}
We reduce from a planar monotone $3$-SAT instance $(U, C)$; refer to
Fig.~\ref{fig:sec:reduction}. We modify its rectilinear representation such that
each vertex and clause rectangle covers sufficiently many cells of a
checkerboard and each edge covers the entire column between its two endpoints.
We place positive literal gadgets in each blue cell of a rectangle corresponding
to a variable.  We place a clause gadget in each positive clause rectangle $R_c$
such that it is aligned with the right-most edge of $R_c$.  The literal gadget
$X$ of a variable $x$ is connected to its corresponding literal gadget $X'$ in
$R_c$ by a placing a literal gadget in each blue cell that is covered by the
$\Gamma$-shape that connects $X$ to $X'$; refer to
Fig.~\ref{fig:sec:reduction:mod}.  Finally, we place a copy gadget in each
orange cell between two literal gadgets of the same variable.  The negative
clauses are obtained by mirroring the modified rectilinear representation
vertically and repeating the construction for the positive clauses. To negate
the state of the variable we place the inverter gadget immediately below a
variable (red cells in Fig.~\ref{fig:sec:reduction:mod}).

\myparagraph{Correctness}
\begin{figure}[tb] \centering \subfloat[]{
		\includegraphics[page=1]{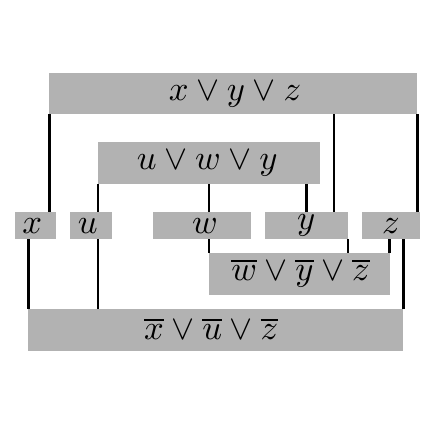} }
		\subfloat[\label{fig:sec:reduction:mod}]{
			\includegraphics[page=4]{fig/rectilinear_representation.pdf} }
			\caption{(a) Planar monotone 3-SAT instance $(U, C)$ with a rectilinear
				representation.  (b) Modified rectilinear representation of $(U, C)$ on
			a checkerboard.} \label{fig:sec:reduction} \end{figure}
 Assume that $(U, C)$ is satisfiable. Depending on whether a variable $u$ is
 true or false, we place all vertices on a positive placement of a positive
 literal gadget and on the negative placement of negative literal gadget of the
 variable.  By Property~\ref{lemma:valid_literal_gadget}, the placement induces
 a $\DiskMap$-framed drawing of all literal gadgets and
 Property~\ref{lemma:valid_copy_inverter_gadget} ensures the copy and inverter
 gadgets have a $\DiskMap$-framed drawing.  Since at least one variable of each
 clause is true, there is a $\DiskMap$-framed drawing of each
 clause gadget by Property~\ref{lemma::valid_clause_gadget}. 
 
 Now consider that the clustered graph $\Cluster$ has a $\DiskMap$-framed
 drawing. Let $X$ and $Y$ be two positive (negative) literal gadgets connected
 with a copy gadget. By Property~\ref{lemma:consistent_copy_gadget}, a drawing
 of $X$ is positive if and only if the drawing of $Y$ is positive.
 Property~\ref{lemma:consistent_copy_gadget}
 ensures that the drawing of a positive literal gadget $X$ is positive if and
 only if the drawing of the negative literal gadget $Y$ is negative, in case
 that both are joined with an inverter gadget.  Further,
 Property~\ref{lemma:no_infeasible_drawing} states that each cycle vertex lies
 either in a positive or a negative region.  Thus, the truth value of a variable
 $u$ can be consistently determined by any drawing of a literal gadget of $u$.
 By Property~\ref{lemma::valid_clause_gadget}, the clause gadget $K$ has no
 $\DiskMap$-framed drawing such that all drawings of the literal gadgets of $K$
 are negative.  Thus, the truth assignment indeed satisfies $C$.

\begin{theorem}
	The problem \textsc{\clusterDrawing} with axis-aligned squares is $\cNP$-hard
	even when the clustered graph $\Cluster$ is restricted to vertex degree 5 and
	the obstacle number of $\DiskMap$ is two.
\end{theorem}


\section{Conclusion}

We proved that every clustered planar graph  with a planar disk arrangement
$\DiskMap$ and a $\DiskMap$-framed embedding $\embedding$ has a
$\DiskMap$-framed straight-line drawing homeomorphic to $\embedding$. If the
requirement of the disk arrangement to satisfy condition (P2) is dropped, we
proved that it is $\cNP$-hard to decide whether $\Cluster$ has a
$\DiskMap$-framed straight-line drawing. We are not aware whether our problem is
known to be in $\cNP$. We ask whether techniques developed by Abrahamsen et
al.~\cite{Abrahamsen:2018:AGP:3188745.3188868} can be used to prove
$\exists\mathbb{R}$-hardness of our problem.

Angelini et al.~\cite{10.1007/978-3-662-45803-7_34} showed that if $\Cluster$ is
not embedded and all squares have the same size, it is $\cNP$-hard to decide
whether $\Cluster$ has a $\DiskMap$-framed drawing. They posed as an open
problem whether the same is true for embedded graphs. In our construction, the
squares have constant number of different side lengths and the side length of
the largest square is 32 time longer then the side length of the smallest rectangle.
We conjecture that our construction can be modified to show that it is indeed
$\cNP$-hard to decide whether a clustered graph $\Cluster$ with a non-planar
arrangement of squares (disk) of unit size and a $\DiskMap$-framed embedding
$\embedding$ has a $\DiskMap$-framed drawing that is homeomorphic to
$\embedding$. Further, we ask whether the obstacle number can be reduced to one.

\bibliography{strings_short,references_no_doi}

\newcommand{\bibdac}[2]{DAC'#2} \newcommand{\bibinvisau}[1]{Proc. Australian
  Symp. Inf. Vis. (invis.au #1)} \newcommand{\bibieeepdp}[2]{Proc. #1 IEEE
  Symp. Par. Distr. Process. #2} \newcommand{\bibieeecs}[1]{Proc. IEEE Int.
  Symp. Circ. Syst. #1} \newcommand{\bibcccg}[2]{CCCG'#2}
  \newcommand{\bibswat}[2]{SWAT'#2} \newcommand{\bibipco}[2]{IPCO'#2}
  \newcommand{\bibsofsem}[2]{SOFSEM'#2} \newcommand{\bibstoc}[2]{STOC'#2}
  \newcommand{\bibfocs}[2]{FOCS'#2} \newcommand{\bibsoda}[2]{SODA'#2}
  \newcommand{\bibgd}[2]{GD'#2} \newcommand{\bibinfovis}[1]{InfoVis'#1}
  \newcommand{\bibvis}[1]{Vis'#1} \newcommand{\bibpvis}[1]{PacificVis'#1}
  \newcommand{\bibsoftvis}[2]{SoftVis'#2} \newcommand{\bibeurocg}[2]{EuroCG'#2}
  \newcommand{\bibsocg}[2]{SoCG'#2} \newcommand{\bibwads}[2]{WADS'#2}
  \newcommand{\bibwg}[2]{WG'#2} \newcommand{\bibgta}{Proceedings of the
  Conference at Graph Theory and Applications}
  \newcommand{\bibisaac}[2]{ISAAC'#2} \newcommand{\bibcocoon}[2]{COCOON'#2}
  \newcommand{\bibtamc}[2]{TAMC'#2} \newcommand{\bibicalp}[2]{ICALP'#2}
  \newcommand{\biblatin}[2]{LATIN'#2} \newcommand{\bibesa}[2]{ESA'#2}
\begin{thebibliography}{10}
\providecommand{\url}[1]{\texttt{#1}}
\providecommand{\urlprefix}{URL }
\providecommand{\doi}[1]{https://doi.org/#1}

\bibitem{Abrahamsen:2018:AGP:3188745.3188868}
Abrahamsen, M., Adamaszek, A., Miltzow, T.: {The Art Gallery Problem is
  $\exists\mathbb{R}$-complete}. In: \bibstoc{50th}{18}. pp. 65--73. ACM (2018)

\bibitem{DBLP:journals/corr/abs-1709-09209}
Akitaya, H.A., Fulek, R., T{\'{o}}th, C.D.: {Recognizing Weak Embeddings of
  Graphs}. In: {Artur Czumaj} (ed.) \bibsoda{Twenty-Ninth}{18}. pp. 274--292.
  SIAM (2018)

\bibitem{JGAA-367}
{Alam}, M., {Kaufmann}, M., {Kobourov}, S.G., {Mchedlidze}, T.: {Fitting Planar
  Graphs on Planar Maps}. J. Graph Alg. Appl.  \textbf{19}(1),  413--440 (2015)

\bibitem{10.1007/978-3-662-45803-7_34}
Angelini, P., {Da Lozzo}, G., {Di Bartolomeo}, M., {Di Battista}, G., Hong,
  S.H., Patrignani, M., Roselli, V.: {Anchored Drawings of Planar Graphs}. In:
  Duncan, C., Symvonis, A. (eds.) \bibgd{22}{14}. pp. 404--415. Springer (2014)

\bibitem{DBLP:journals/dcg/AngeliniFK11}
Angelini, P., Frati, F., Kaufmann, M.: {Straight-Line Rectangular Drawings of
  Clustered Graphs}. Disc. \& Comput. Geom.  \textbf{45}(1),  88--140 (2011)

\bibitem{10.1007/978-3-319-62127-2_7}
Banyassady, B., Hoffmann, M., Klemz, B., L{\"o}ffler, M., Miltzow, T.:
  {Obedient Plane Drawings for Disk Intersection Graphs}. In: Ellen, F.,
  Kolokolova, A., Sack, J.R. (eds.) \bibwads{15th}{17}. pp. 73--84. Springer
  (2017)

\bibitem{DBLP:journals/ijcga/BergK12}
de~Berg, M., {Khosravi}, A.: {Optimal Binary Space Partitions for Segments in
  the Plane}. Int. J. Comput. Geom. \& Appl.  \textbf{22}(3),  187--206 (2012)

\bibitem{BLASIUS2016306}
Bläsius, T., Rutter, I.: {A New Perspective on Clustered Planarity as a
  Combinatorial Embedding Problem}. Theoretical Comput. Sci.  \textbf{609}(2),
  306 -- 315 (2016)

\bibitem{CORTESE20091856}
Cortese, P.F., {Di Battista}, G., Patrignani, M., Pizzonia, M.: {On Embedding a
  Cycle in a Plane Graph}. Disc. Math.  \textbf{309}(7),  1856 -- 1869 (2009)

\bibitem{Eades2006}
Eades, P., Feng, Q., Lin, X., Nagamochi, H.: {Straight-Line Drawing Algorithms
  for Hierarchical Graphs and Clustered Graphs}. Algorithmica  \textbf{44}(1),
  1--32 (2006)

\bibitem{DBLP:conf/esa/FengCE95}
Feng, Q., Cohen, R.F., Eades, P.: {Planarity for Clustered Graphs}. In:
  Spirakis, P. (ed.) \bibesa{Third}{95}. pp. 213--226. Springer (1995)

\bibitem{10.1007/3-540-58950-3_377}
Godau, M.: {On the Difficulty of Embedding Planar Graphs with Inaccuracies}.
  In: Tamassia, R., Tollis, I.G. (eds.) \bibgd{Second}{94}. pp. 254--261.
  Springer (1995)

\bibitem{ribo2006realization}
{Rib{\'o} Mor}, A.: {Realization and Counting Problems for Planar Structures}.
  Ph.D. thesis, FU Berlin (2006),
  \url{https://refubium.fu-berlin.de/handle/fub188/10243}

\end{thebibliography}

\newpage

\iftrue
\appendix

\section{Drawing on Non-Planar Disk Arrangements }
\label{apx:hardness}

We study the following problem referred to as \textsc{\clusterDrawing}. Given a
planar clustered graph $\Cluster$ with an embedding $\embedding$ and a
non-planar disk arrangement $\DiskMap$, is there are $\DiskMap$-framed drawing
$\Gamma$ that is homeomorphic to $\embedding$ and $\DiskMap$. This requires that
in the transition from $\embedding$ to $\Gamma$ at any point in time an edge
$uv$ does not intersect a geometric object other than its own clusters.  Note
that if the disks $\DiskMap$ are allowed to overlap and $\Cluster$ is the
intersection graph of $\DiskMap$, the problem is known to be
$\cNP$-hard~\cite{10.1007/978-3-319-62127-2_7}.  Thus, in the following we
require that the disk may not overlap, but there can be disk-pipe intersection.
By Alam at al.~\cite{JGAA-367} it follows that the problem restricted to thin
touching rectangles instead of disks is $\cNP$-hard.  We strengthen this result
and prove that in case that the rectangles are axis-aligned squares and not
allowed to touch the problem remains $\cNP$-hard. The squares in the proof can
be replaced in a straight-forward manner by disks.

To prove $\cNP$-hardness for \clusterDrawing problem we reduce from
\textsc{Planar Monotone 3-SAT}~\cite{DBLP:journals/ijcga/BergK12}. For each
literal and clause we construct a graph $\Cluster$ with a disk arrangement
$\DiskMap$ of $\Cluster$ such that each disk contains exactly one vertex.  We
refer to these instances as \emph{literal} and \emph{clause gadget}. In order to
transport information from the literals to the clauses, we construct a
\emph{copy} and \emph{inverter gadget}.  The design of the gadgets is inspired
by Alam et al.~\cite{JGAA-367}, but due to the restriction to squares rather
than rectangles, requires a more careful placement of the geometric objects.

\paragraph{Notation}
A line $l$ separates the euclidean plane in two \emph{half planes $h_1$ and
$h_2$} that are \emph{spanned by~$l$}. Way say that $l$ \emph{supports $h_1$
($h_2$)}.
 
\subsection{Regulator}

\begin{figure}[tb]
	\centering
	\includegraphics[page=2]{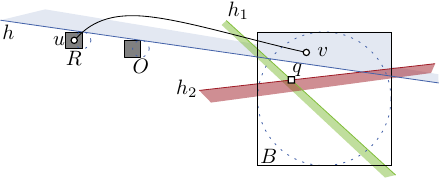}
	\caption{Regulator}
	\label{fig:regulator}
\end{figure}

Let $B$ be  an axis-aligned square that contains a vertex $v$ in its interior
and let $h_1, h_2$ be two half planes such that the intersection $q$ of their
supporting lines $l_1, l_2$ lies in the interior of $B$. We say that $h_1$ and
$h_2$ are \emph{proper half planes of $B$}.  We describe the construction of a
gadget that restricts the feasible placements of $v$ in a $\DiskMap$-framed
drawing by a half plane $h$ that excludes a placement of $v$ in
$h_1 \cap h_2$ but allows for a placement in $h_1 \cap B$ or $h_2 \cap B$.
Since $q$ lies in the interior of $B$, there is a half plane $h$ that does not
contain $q$ and for each $i=1,2$, $h \cap h_i \cap B$ is non-empty.

We construct a \emph{regulator gadget of $v$ in $B$ with respect to $h_1$ and
$h_2$} as follows.  Let $l$ be the supporting line of $h$.  We create 
two axis-aligned squares $R$ and $O$ such that $R, O$ and $B$ intersect $l$ in
this order and $h$ neither intersects the interior of $R$ nor the interior of
$O$. Place a vertex $u$ in $R$ and route an edge $uv$ through $h \cup R \cup B$.

\begin{lemma} \label{lemma:regulator}	
	Let $W$ be a regulator gadget of $v$ in $B$ with respect to two proper half
	planes $h_1$ and $h_2$.  For every point $p_v \in h \cap B$ there is a
	$\DiskMap$-framed drawing $\Gamma$ such that $v$ lies on $p_v$. There is no
	$\DiskMap$-framed drawing of $W$ such that $v$ lies in~$\overline{h} \cap B$.
\end{lemma}
	
\begin{proof} 
	By construction of $W$, there is for every point $p_v \in h \cap B$ a
	$\DiskMap$-framed drawing $\Gamma$ such that $v$ lies on $v$. 

	The supporting line $l$ of $h$ intersects the boundary of $R$ 
	and does not intersect the interior of $O$. Let $r$ and $o$ be points in the
	intersection of  $l$ with $R$ and $O$, respectively. Since $\Gamma$ is
	homeomorphic to $W$ the edge $uv$ intersects $l$ on the ray starting in $o$ in
	the direction towards $r$. Therefore, $u$ and $v$ lie on different sides of
	$l$. Since $u \in R$, it follows that $v \in \overline{h}$.
\end{proof}

We refer to the intersection $h \cap B$ as the \emph{regulated region of $v$  in
$B$}. Thus, by the construction of $W$, the regulated region $Q$ has a non-empty
intersection with $h_1 \cap B$ and $h_2 \cap B$. Thus, by the lemma there is for
each placement of $v$ in $Q \cap h_i \cap B, i=1,2,$ a $\DiskMap$-framed drawing. On the
other hand, since $h \cap h_1 \cap h_2 \cap B = \emptyset$, there is no
$\DiskMap$-framed drawing such that $v$ lies in $h_1 \cap h_2 \cap B$.

\subsection{Literal Gadget}
\label{sec:literal_gadget}

\begin{figure}[t]
	\centering
	\subfloat[]{
		\includegraphics[page=1, height=0.5\textwidth]{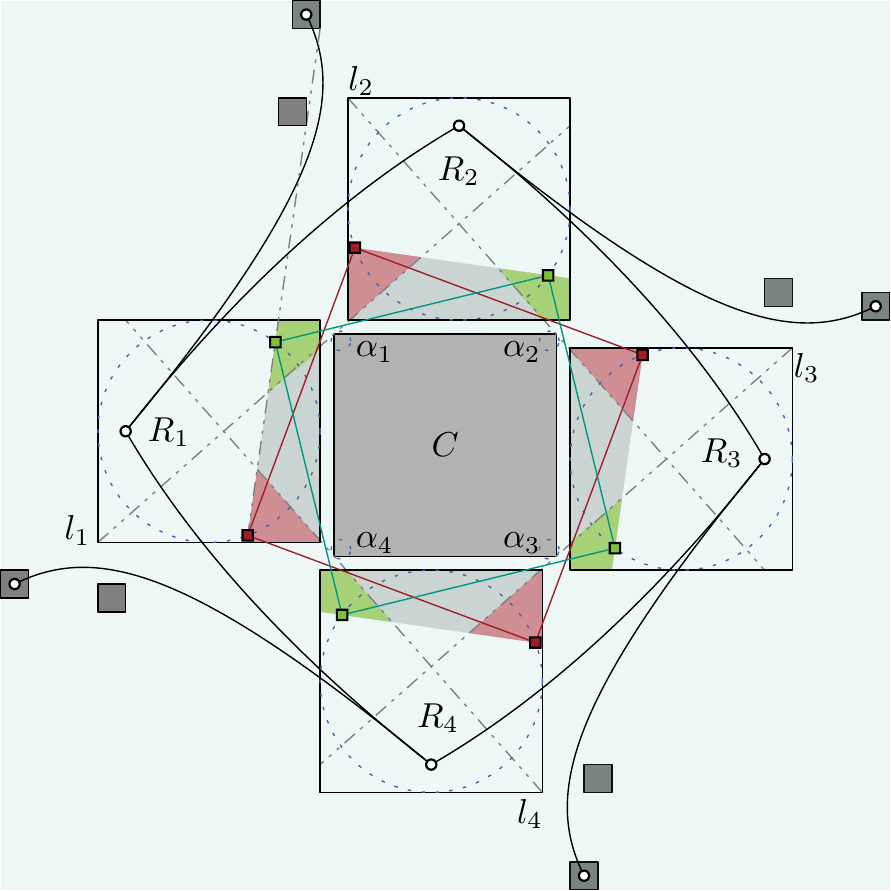}
	}
	\subfloat[\label{fig:hardness:literal:constr}]{
		\includegraphics[page=3, height=0.5\textwidth]{fig/variable.pdf}
	}
	\caption{Literal gadget}
	\label{fig:hardness:literal}
\end{figure}

In this section we construct a clustered graph $\Cluster$ with an arrangement
of squares $\DiskMap$ that models a literal $u$. The \emph{positive literal
gadget} is depicted in Fig.~\ref{fig:hardness:literal}. We obtain the
\emph{negative literal gadget} by mirroring vertically.

The \emph{center block} is a unit square $C$ with corners $\alpha_1, \alpha_2,
\alpha_3, \alpha_4$ in clockwise order.  For each corner $\alpha_i$ of $C$
consider a line $l_i$ that is tangent to $C$ in $\alpha_i$, i.e, $l_i \cap C =
\{\alpha_i\}$.  Let $p_i$ be the intersection of lines $l_{i-1}$ and $l_{i}$
where $l_0 = l_4$; refer to Fig.~\ref{fig:hardness:literal:constr}.  Let $R_1,
\dots, R_4$ be four pairwise non-intersecting squares that are disjoint from $C$
such that $R_i$ contains $p_i$ in its interior.  We add a cycle
$v_1v_2v_3v_4v_1$ such that $v_i \in R_i$. We refer to the vertices $v_i$ as
\emph{cycle vertices} of the \emph{cycle block} $R_i$. For each $i$, let $j_i$
be a half plane that contains $R_{i+1}$ but does not intersect $C$.  We place a
regulator $W_i$ of $v_i$  with respect to $h_{i-1}$ and $h_{i}$ and position it
such that it lies in $j_i$, where $h_i$ is the half plane spanned by $l_i$ with
$C \not\subseteq h_i$.  This finishes the construction.

We now show that there exist two combinatorially different realizations. 
Consider $R_1$ and
its two adjacent squares $R_4$ and $R_2$.  Let $Q_i$ be the regulated region of
$R_i$ with respect to $W_i$.  Then the intersection $I_1:= \overline{h_4} \cap
\overline{h_2} \cap Q_1\not=\emptyset$. We refer to $I_1$ as the \emph{infeasible region of $R_1$}.
The intersection $h_4 \cap Q_1$ is the \emph{positive region $P_1$ of
$R_1$}.  The region $h_2 \cap Q_1$ is the \emph{negative region $N_1$
of $R_1$}.  All these regions are by construction not empty. The
positive, negative and infeasible region of $R_i, i \not= 1$ are defined
analogously.

\rephrase{Property}{\ref{lemma:no_infeasible_drawing}}{
	\PropertyNoInfeasibleDrawing
}

\begin{proof}
	Consider a $\DiskMap$-framed
	drawing $\Gamma$ with an edge $v_iv_{i+1}$ such that $v_i$ lies in
	$\overline{P_i}$. We show that $v_{i+1}$ lies in $N_{i+1}$.  If $v_{i+1}$ lies
	in $\overline{h_{i}}$, then $v_i$ and $v_{i+1}$ lie on the same side of $l_{i}$,
	which is tangent to $\alpha_{i+1}$. Thus, $v_iv_{i+1}$ intersects $C$. If
	follows that $v_{i+1}$ lies in $h_{i}$ and therefore in the negative
	region $N_{i+1}$.
	
	Assume that $v_1$ lies in its infeasible region $I_1$, then $v_2$ lies in
	$N_2$ by the above observation. Likewise, $v_3, v_4, v_1$ lie in $N_3, N_4,
	N_1$, respectively. This contradicts $N_1 \cap I_1 = \emptyset$. This
	generalizes to all $v_i, i \not =1$. Thus, each $v_i$ either lies in $P_i$ or
	in $N_i$.  Moreover, if one $v_i$ lies in $N_i$ the above observation yields
	that all of them lie in their negative region.
\end{proof}

In the following, we fix the placement of the cycle blocks for literal gadgets
as depicted in Fig.~\ref{fig:hardness:literal}. This allows us to show that the
literal gadget has a $\DiskMap$-framed drawing where all cycle vertices lie in
their positive region and one where all cycle vertices lie in their negative
region. We refer to the former as \emph{positive} and the latter as
\emph{negative} drawing.  We construct two specific drawings.  Let $D_i$ be the circle
inscribed in the square $R_i$. Since $P_i$ and $N_i$ are obtained from the
intersection of two half planes with $R_i$, they are convex.  The intersection
$p^+_i$ of the boundary of $P_i$ with $D_i$ that does not lie on the boundary of
$Q_i$ is the \emph{positive placement of $v_i$}.  Analogously, we obtain the
\emph{negative placement $p^-_i$ of $v_i$}. The positive and negative placement
induce two straight-line drawings of the graph induced by cycle vertices. By
Lemma~\ref{lemma:regulator} we can extend these drawings to $\DiskMap$-framed
drawings of the whole literal gadget, including the regulators.

\rephrase{Property}{\ref{lemma:valid_literal_gadget}}{
	\PropertyValidLiteral
}


\subsection{Copy and Inverter Gadget}

\begin{figure}[t]
	\centering
		\includegraphics[page=1, width=.99\textwidth]{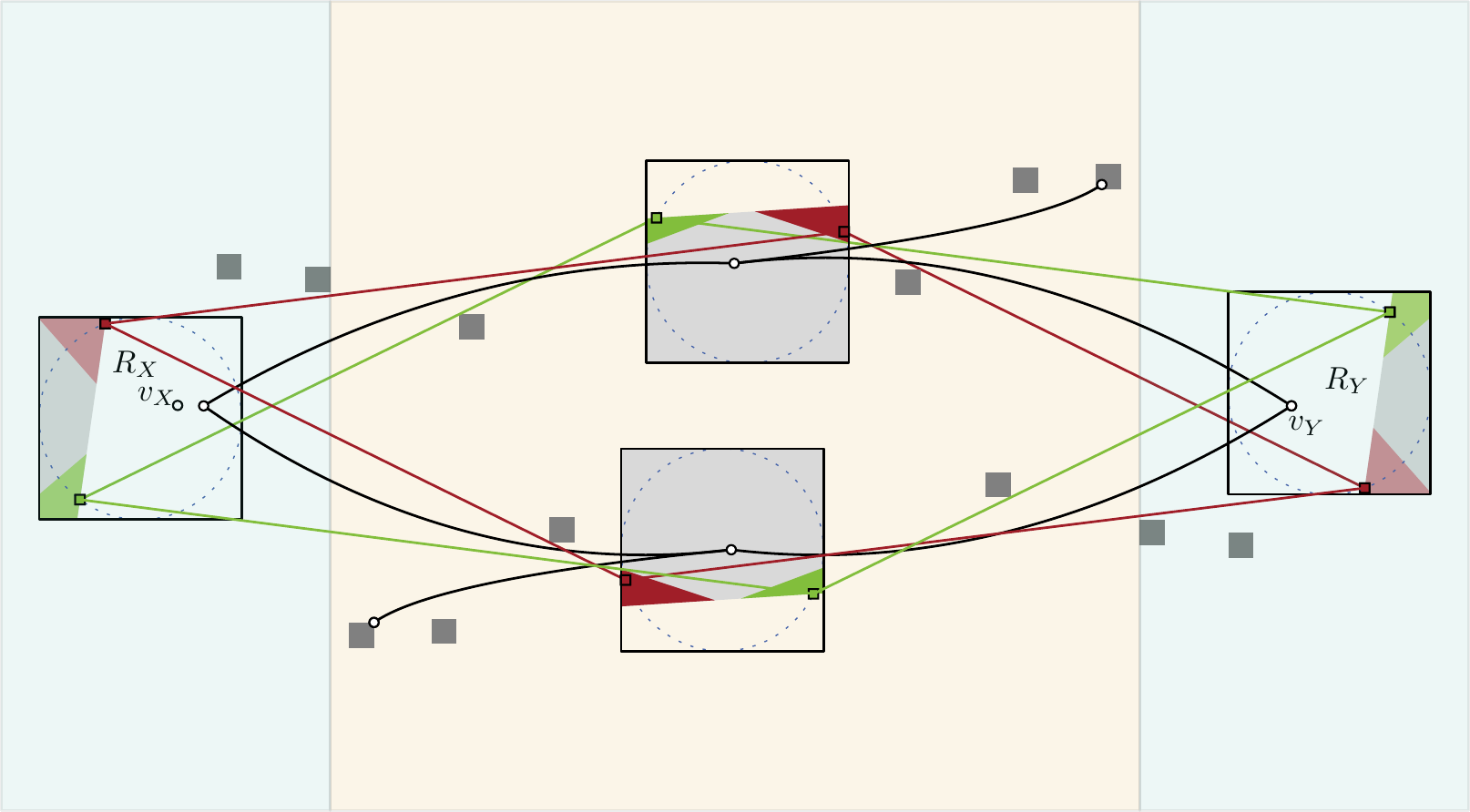}
	\caption{Copy gadget. Green and red regions depict positive and negative
	regions, respectively.}
	\label{fig:hardness:copy}
\end{figure}

In this section, we describe the copy and inverter gadget; see
Fig.~\ref{fig:hardness:copy}.  The copy gadget connects two positive or
two negative literal gadgets $X$ and $Y$ such that a drawing of $X$ is positive
if and only if the drawing of $Y$ is positive. Correspondingly, the inverter
gadget connects a positive literal gadget $X$ to a negative literal gadget $Y$
such that the drawing of $X$ is positive if and only if the drawing of $Y$ is
negative.  The construction of the inverter and the copy gadget are symmetric.

A copy gadget of two negative literal gadgets is obtained by vertically
mirroring the copy gadget that connects two positive literals. Correspondingly,
we obtain an inverter gadget that connects a negative literal with a positive
literal by mirroring the inverter gadget that connects a positive literal with a
negative literal. Thus, in the following we describe only the construction of
the copy gadget with two positive literals. Whenever necessary we emphasize
differences for the inverter gadget.

Let $X$ and $Y$ be two positive literal gadgets whose center blocks are aligned
on the $x$-axis with a sufficiently large distance. We construct the copy gadget
that \emph{connects} $X$ and $Y$ as follows.  Let $R_X$ and $R_Y$ be the two
cycle blocks of the literal gadgets $X$ and $Y$, respectively, with minimal
distance on the $x$-axis. For $Z \in \{X, Y\}$, let $P_Z$ and $N_Z$ be the
positive and negative regions of $R_Z$.  Since $N_Z$ an $P_Z$ are convex and
their intersection is empty, there exist a half plane $h_Z$ that contains $N_Z$
but not $P_Z$, and vice versa.  We refer to $h_Z$ as \emph{positive half plane
of $Z$} if it contains the positive region $P_Z$, otherwise it is
\emph{negative}.

\begin{figure}[t]
	\centering
	\includegraphics[page=3, width=.99\textwidth]{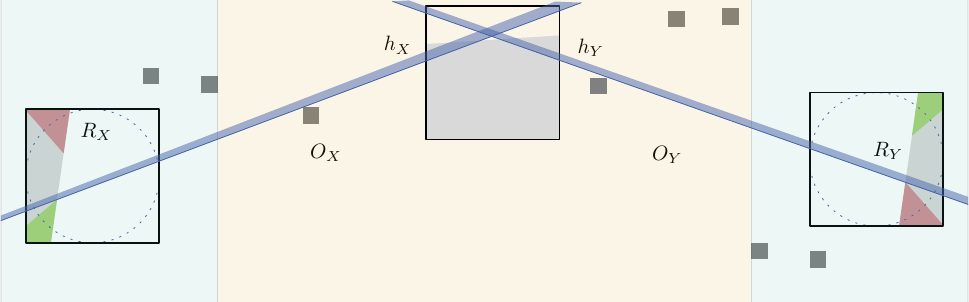}
	\caption{Construction of the square $B$ of the copy gadget.}
	\label{fig:hardness:copy:construction}
\end{figure}

Consider a negative half plane $h_X$ of $X$ and a positive half plane $h_Y$ of $Y$;
refer to Fig.~\ref{fig:hardness:copy:construction}. We create two
non-intersecting squares $O_X$ and $O_Y$ that are contained in the intersection
of $\overline{h_X}$ and $\overline{h_Y}$ such that a corner of $O_Z$ lies on the
supporting line of $h_Z, Z \in \{X,Y\}$. Let $I$ be the intersection of the
supporting lines of $h_X$ and $h_Y$. We place a square $B$ with a vertex $b$ and
the intersection $I$ in its interior.  Additionally, we add a regulator of $b$
with respect to $h_X$ and $h_Y$ to exclude the intersection $h_X\cap h_Y$ as
feasible placement of $b$.  We route the edges $bv_X$ and $bv_Y$ through $R_X
\cup h_X \cup B$ and $R_Y \cup h_Y \cup B$ respectively.  This construction
ensures that in a $\DiskMap$-framed drawing a placement of the vertex $v_X$ in
the negative region $N_X$ excludes the possibility that the vertex $v_Y$ lies in
the positive region $P_Y$.  In order to ensure that $v_X$ cannot lie at the same
time in $P_X$ as $v_Y$ in $N_Y$, we construct a square $B'$ with respect to a
positive half plane of $X$ and a negative half plane of $Y$ analogously to $B$.
If the distance between $X$ and $Y$ is sufficiently large, we can ensure the
intersection of $B$ and $B'$ is empty.  In the construction of the inverter
gadget one has to consider two positive half planes and two negative half
planes, respectively.  We refer to the corresponding gadgets as \emph{copy} and
\emph{inverter gadget}. We say that the copy and inverter gadget \emph{connect}
two literals.

\rephrase{Property}{\ref{lemma:consistent_copy_gadget}}{
	\PropertyConsistentCopyGadget
}

\begin{proof}
	
	By Property~\ref{lemma:valid_literal_gadget} the vertices $v_X$ and $v_Y$ of
	$X$ and $Y$ cannot lie the infeasible regions of $X$ and $Y$, respectively.
	Thus, similar to the proof of Lemma~\ref{lemma:no_infeasible_drawing} we can
	assume for the sake of contradiction that the vertex $b$ of block $B$ lies in
	the intersection of $\overline{h_X}$ and $\overline{h_Y}$. Thus, vertex $v_X$
	lies in the negative region of $R_X$ and $v_Y$ in the positive region of
	$R_Y$. But then vertex $b'$ of the block $B'$ lies in $h'_X$ and $h'_Y$.  This
	contradicts that $Q' \cap h'_X \cap h'_Y \cap B'$ is empty, where $Q'$ is the
	regulated region of $B'$, and $h'_X$ and $h'_Y$ are the negative and positive
	half planes of $R_X$ and $R_Y$, respectively.	
\end{proof}

\begin{figure}[t]
	\centering
		\includegraphics[page=2, width=.99\textwidth]{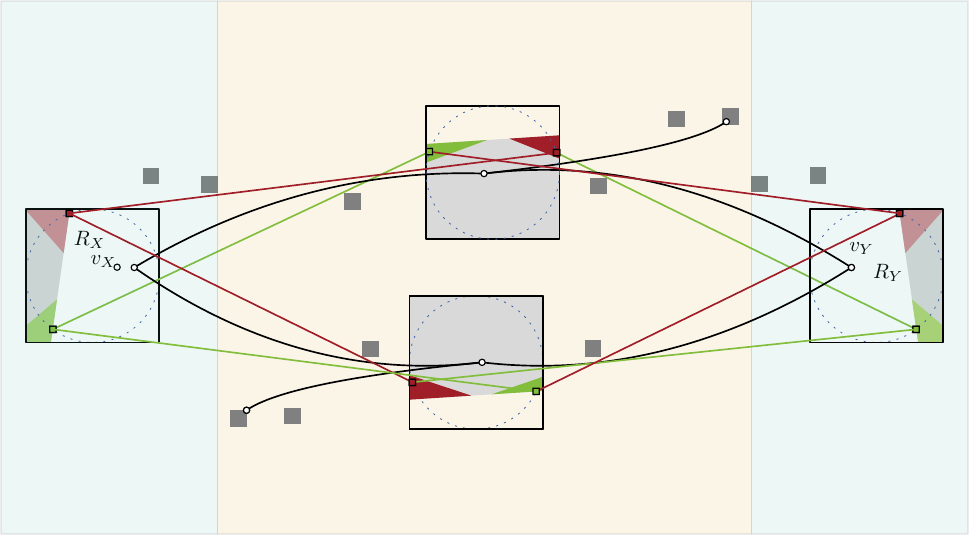}
	\caption{Inverter gadget. Green and red regions depict positive and negative
	regions, respectively.}
	\label{fig:hardness:inverter}
\end{figure}

The same argumentation is applicable to the inverter gadget.

\begin{property}
	\label{lemma:consistent_inverter_gadget}
	%
	Let $\Gamma$ be a $\DiskMap$-framed drawing of a positive literal $X$ and a
	negative literal $Y$ connected by an inverter gadget. Then the $\DiskMap$-framed
	drawing of $X$ in $\Gamma$ is positive if and only if the $\DiskMap$-framed
	drawing of $Y$ is negative.
\end{property}

From now on we refer to the exact placement of the squares as depicted in
Fig.~\ref{fig:hardness:copy} and Fig.~\ref{fig:hardness:inverter} as the copy
and inverter gadget, respectively. The positive and negative placement of the
literal gadgets induce a $\DiskMap$-framed drawing of the copy gadgets as
indicated by the green and red straight-line segments, respectively. By
Lemma~\ref{lemma:regulator} we can extend these drawings to drawings of the
entire gadget.

\rephrase{Property}{\ref{lemma:valid_copy_inverter_gadget}}{
	\PropertyValidCopyInverter
}

\subsection{Clause Gadget}
\label{sec:clause}
\begin{figure}[tb]
	\centering
		\includegraphics[page=1, width=.99\textwidth]{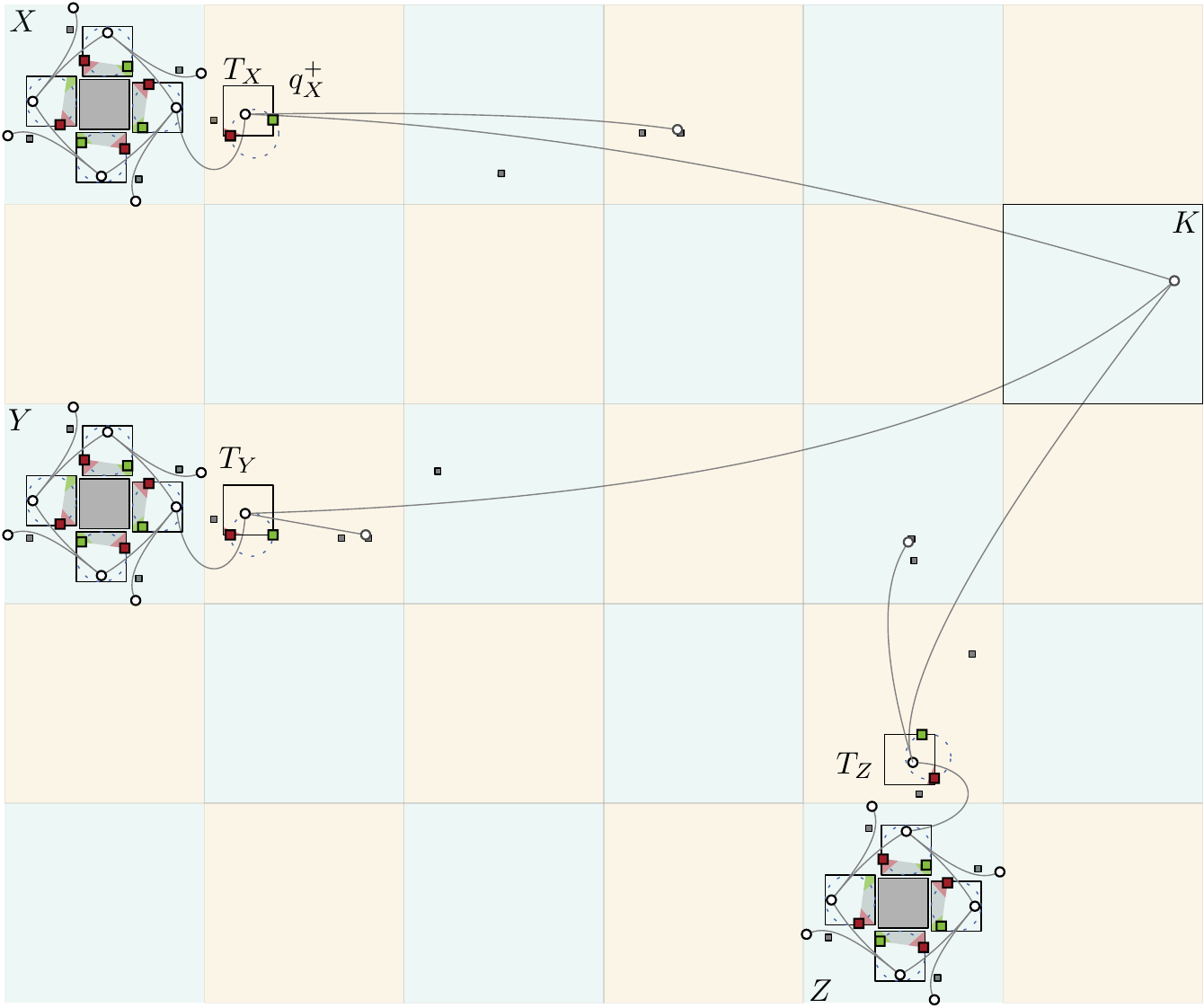}
		\caption{Clause gadget.}
	\label{fig:hardness:clause}
\end{figure}


We construct a \emph{clause gadget} with respect to three positive literal
gadgets $X, Y,Z$ arranged as depicted in Fig.~\ref{fig:hardness:clause}. The
negative clause gadget, i.e., a clause with three negative literal gadgets, is
obtained by mirroring vertically.

We construct the clause gadget in two steps. First, we place a \emph{transition
block} $T_A$ close to each literal gadget $A \in \{X, Y, Z\}$. In the second
step, we connect the transition block to a vertex $k$ in a \emph{clause block
$K$} such that for every placement of $k$ in $K$ at least one drawing of the
literal gadgets has to be positive.

\begin{figure}[tb]
	\centering
	\includegraphics[page=2]{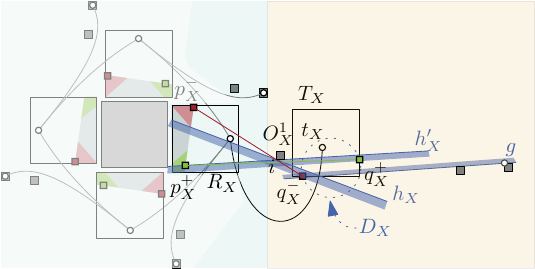}	
	\caption{Construction of the transition block.}
	\label{fig:transition_block}
\end{figure}

Consider the literal gadget $X$ and let $R_X$ be the right-most cycle block of $X$.  Let
$h_X$ be a negative half plane of $R_X$, i.e., $h_X$ contains the positive
region $P_X$ but not the negative region $N_X$, refer to
Fig.~\ref{fig:transition_block}. We now place a transition block $T_X$ such that
the intersection $T_X \cap h_X$ has small area.  Further, let $p^+_X$ and
$p^-_X$ be the positive and negative placements of $X$, respectively. Let
$q^-_X$ be a point in $T_X \cap h_X$. Let $i$ be the intersection point of the
supporting line $l_X$ of $h_X$ and the line segment $p^-_Xq^-_X$. We place an
obstacle $O^1_X$ such that $l_X$ is tangent to $O^1_X$ in point~$i$.
Finally, we place a \emph{transition vertex} $t_X$ in the interior of $T_X$ and
route the edge $v_Xt_X$ through $h_X \cup T_X \cup R_X$, where
$v_X \in R_X$.

Consider a half plane $h'_X$  such that $O^1_X \not\subseteq h'_X$ and $N_X
\not\subseteq h'_X$ and such that the supporting line $l'_X$ of $h'_X$ contains
$p^+_X$ and is tangent to $O^1_X$. Let $q^+_X$ be a point $h'_X \cap R_X$.
Observe that $q^+_X$ and $q^-_X$ induce a positive and negative drawing of $X$,
respectively.  Further, if $X$ has a negative drawing, then $t_X$ lies in the
region $h_X \cap T_X$. In the following, we refer to $h_X \cap T_X$ as the
\emph{negative region} of $T_X$. The transition blocks of $Y$ and $Z$ are
constructed analogously with only minor changes. The transition block $T_Z$ of
$Z$ is constructed with respect to the top-most cycle block. Note that we can
choose the points $q^+_A, A \in \{X, Y, Z\}$ independent from each other as long
as each of them induces a positive drawing of the literal gadget $A$.

The blue dotted circle in Fig.~\ref{fig:transition_block} indicates how to
replace the square by a disk $D$ that contains $q^-_X$ and $q^+_X$.  Denote by
$\xmax S$ the maximum $x$-coordinate of a point in a bounded set $S \subset
\R^2$.  Note that $\xmax D\cap h_X > \xmax T_X \cap h_X$. To ensure that our
construction remains correct for disks we add a regulator $R$ with a respect a
half plane $g$ such that $\xmax D\cap h_x \cap g = \xmax T_X \cap h_x \cap g$
and $g$ contains $q^+_X, q^-_x$.

\begin{figure}[tb]
	\centering
	\includegraphics[page=6]{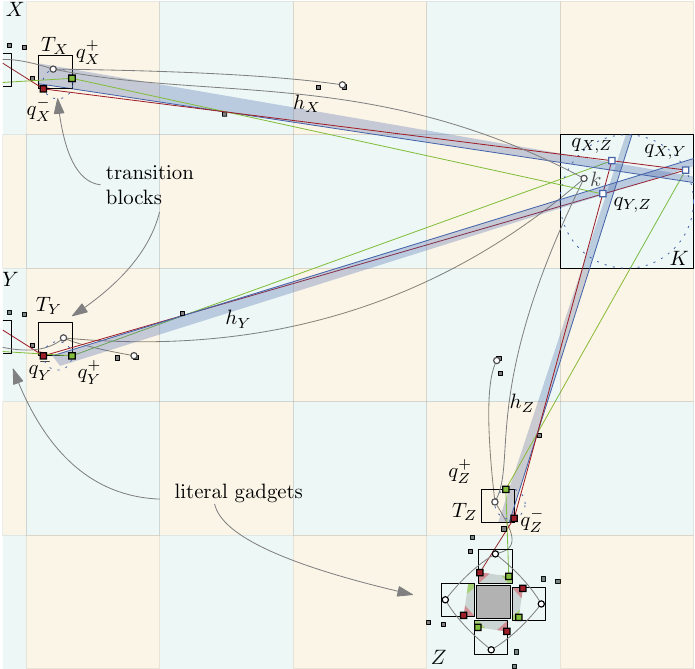}	
	\caption{Construction of the clause block.}
	\label{fig:clause_block}
\end{figure}

Given the placement of the transition block $T_X, T_Y$ and $T_Z$ as depicted in
Fig.~\ref{fig:clause_block}, we construct the \emph{clause block $K$} as
follows. We choose a point $q_{X,Y}$. Let $l^-_X$ and $l^-_Y$ be the lines
through the points $q^-_X, q_{X, Y}$, and $q^-_Y, q_{X, Y}$, respectively.
Further, consider a line $l^-_Z$ with $q^-_Z  \in l^-_Z$ such that the
intersection point $q_{A, Z} := l^-_Z \cap l^-_A, A \in \{X, Y\}$ lies in
between $q^-_A$ and $q_{X, Y}$. Further, let $l^+_X$ be the line through $q^+_X,
q_{Y, Z}$, $l^+_Y$ the line through $q^+_Y, q_{X, Z}$, and let $l^+_Z$ be the line through
$q^+_Z$ and $q_{X, Y}$. Let $h_A$ be a half plane that does not contain the
negative region $N_A$ and whose supporting line contains the intersection $i_A$
of $l^-_A$ and $l^+_A$. We place obstacles $O^2_A$ such that $O^2_A
\not\subseteq h_A$ and the supporting line of $h_A$ is tangent to $O^2_A$ in
point $i_A$. We place the clause box $K$ such that it contains $q_{X, Y},
q_{Y,Z}$, $q_{X,Z}$ and a new vertex $k$ in its interior. We finish the
construction by routing the edges $kt_A$ through $K \cup h_A \cup T_A, A \in
\{X, Y, Z\}$, where $t_A \in T_A$.

By construction we have that for each $y \in \{q^-_Y, q^+_Y\}$ and $z \in
\{q^-_Z, q^+_Z\}$ the points $y, z$ and $q_{Y, Z}$ induce a $\DiskMap$-framed
drawing. The analog statement for the points $q_{X, Z}$ and $q_{X, Y}$ is also
true.  Further, if $h_X \cap h_Y \cap  h_Z = \emptyset$, then there is no
$\DiskMap$-framed drawing such that each vertex $t_A$ lies on $q^-_A$.
Fig.~\ref{fig:clause_block} shows that there is an arrangement of the clause
block and the obstacles such that $h_X \cap h_Y \cap  h_Z$ indeed is empty. From
now on we refer the arrangement in Fig.~\ref{fig:hardness:clause} as the clause
gadget.

\rephrase{Property}{\ref{lemma::valid_clause_gadget}}{
	\PropertyValidClause
}






\subsection{Reduction}

\begin{figure}[tb]
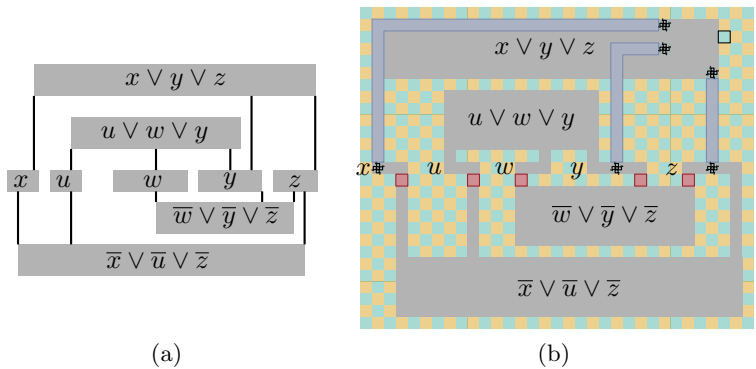

	\centering
	 \subfloat[]{
		\includegraphics[page=1]{fig/rectilinear_representation.pdf}
	}
	\subfloat[\label{fig:reduction:mod}]{
		\includegraphics[page=4]{fig/rectilinear_representation.pdf}
	}
	\caption{Example of planar monotone $3$-SAT instance with a corresponding rectilinear
	representation.}
	\label{fig:3SAT}
\end{figure}

A 3-SAT instance $(U, C)$ on a set $U$ of $n$ boolean variables and $m$ clauses
$C$ is \emph{monotone} if each clause either contains only positive or only
negative literals. It is \emph{planar} if the bipartite graph $G_{U, C} = (U
\cup C, \{ uc \mid u  \in c \text{ or } \overline{u} \in c \text{ with } u \in U
\text{ and } c \in C\})$ is planar. A \emph{rectilinear representation} of a
monotone planar 3-SAT instance is a drawing of $G_{U, C}$ where each vertex is
represented as an axis-aligned rectangle and the edges are vertical line
segments touching their endpoints; see Fig.~\ref{fig:3SAT}. Further, all vertices corresponding to a
variable lie on common line $l$, the positive and negative clauses are separated
by $l$. The problem \textsc{Monotone Planar 3-SAT} asks whether a monotone
planar 3-SAT instance with a given rectilinear representation is satisfiable. De
Berg and Khosravi~\cite{DBLP:journals/ijcga/BergK12} proved that
\textsc{Monotone Planar 3-SAT} is \cNP-complete.  We use this problem to show
that the \textsc{\clusterDrawing} problem is $\cNP$-hard.

\begin{figure}[tb]
	\centering
	\includegraphics[page=2,width=0.99\textwidth]{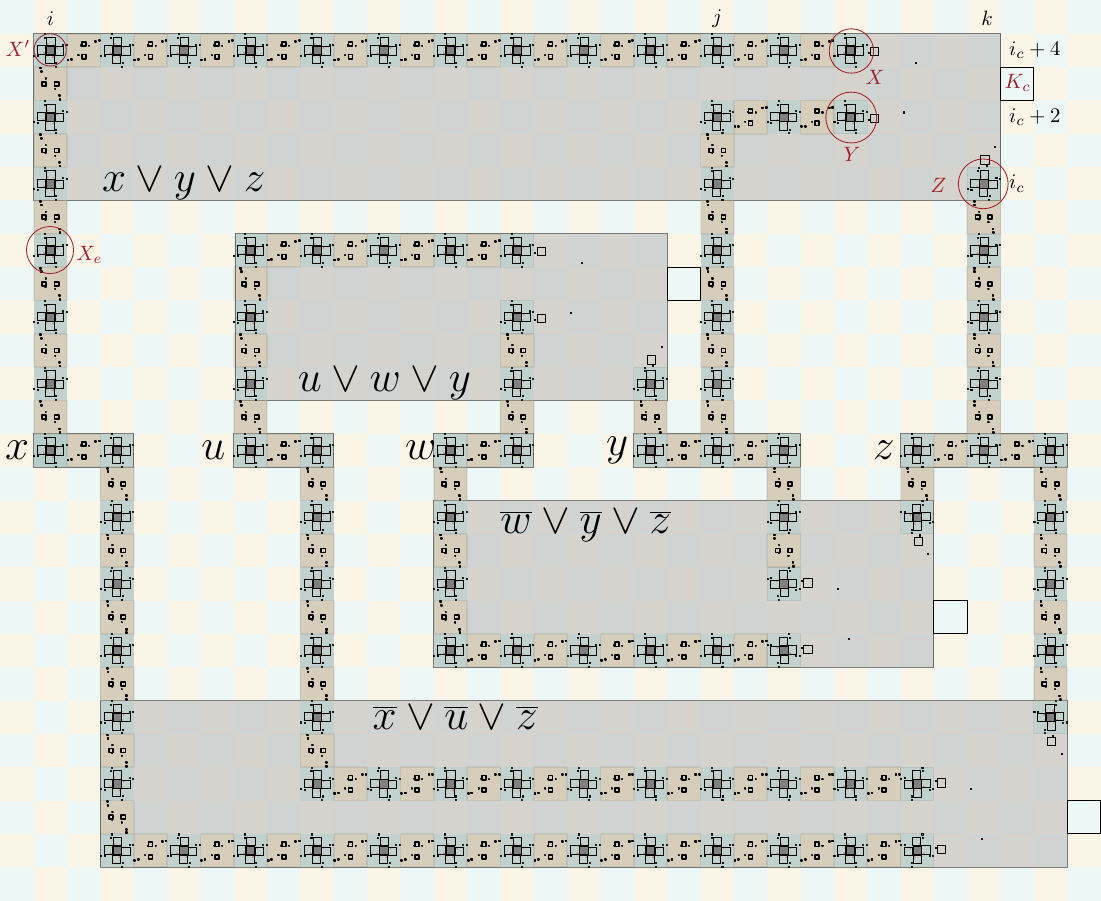}
	\caption{Sketch of the final arrangement of squares. Edges are omitted from the
	drawing.}
	\label{fig:rectilinear_rep}
\end{figure}

\begin{theorem}
	The problem \textsc{\clusterDrawing} with axis-aligned squares is $\cNP$-hard
	even when the clustered graph $\Cluster$ is restricted to vertex degree 5 and
	the obstacle number is $2$.
\end{theorem}

\begin{proof}
 Let $(U, C)$ be a planar monotone 3-SAT instance with a rectilinear
 representation $\Pi$.  Let $l$ be a horizontal or vertical line that intersects
 $\Pi$. Further, let $l_\epsilon$ be a tunnel of width $\epsilon$ around $l$. We
 obtain from $\Pi$ a new rectilinear representation by increasing the width of
 $l_\epsilon$ by an arbitrary positive factor $x$. This operations allows us to
 do the necessary modifications.
 
 In the following we modify $\Pi$ to fit on a checkerboard of $O(|C|)$ rows and
 columns where each column has width $d$ and every row has height $d$. A row or
 column is \emph{odd} if its index is an odd number, otherwise it is
 \emph{even}. The pair $(i, j)$ refers to the cell in column $i$ and row $j$.
 We align all vertices corresponding to variables in the rectilinear
 representation in a common row and such that the left-most variable vertex is
 in column $1$; refer to Fig.~\ref{fig:rectilinear_rep}.  The width of each
 rectangle $r_u$ of variable $u$ is increased to cover $2\cdot n_u - 1$ columns where
 $n_u$ is the number of occurrences of $u$ and $\bar{u}$ in $C$.  To ensure that each $r_u$
 starts in an odd column, we increase the distance between two consecutive
 variables such that the number of columns between the variables is odd and is
 at least three. Since we are able to add an arbitrary number of columns between
 two consecutive variables, we can assume without loss of generality that no two
 edges of the rectilinear representation share a column and that their columns
 are odd. We adapt the rectangle of a clause such that it covers five rows and
 at least six columns, and such that its left and right sides are aligned with
 the left-most and right-most incoming edges, respectively. Note that the
 positive clauses lie in rows with a positive index and the negative clauses in
 rows with a negative index.  Each operation adds add most a constant number of
 columns and rows per vertex and per edge to the layout. Thus, the width and
 height of the final layout is in $O(|C|)$.  Further, it can be computed in time
 polynomial in $|C|$.

 In the following we construct a planar embedded graph $\Cluster$ and an
 arrangement of squares $\DiskMap$ of $\Cluster$.  We use the modified rectilinear layout
 to locally replace the variable by a sequence of positive and negative literals
 connected by either a copy or an inverter gadget. Clauses are replaced with the
 clause gadget and then connected with a sequence of literals and copy gadget to
 the respective literal in the variable.
 
 Observe that the literal gadget is constructed such that all its squares fit in
 a larger square $S$. The copy and inverter gadget together with two literals is
 constructed such that they fit in rectangle three times the size of $S$. The
 clause gadget fits in a rectangle of width six times the size of the square $S$
 and its height is five times the height of $S$.

 Thus, we assume that the size of the square $S$ and the size of the squares of
 the checkerboard coincide. Let $r$ be the row that contains the variable
 vertices. Every column contains at most one edge of the rectilinear
 representation. Thus, we place a positive literal gadget in cell $(i, r)$ if
 the edge in column $i$ connects a variable $u$ to a positive clause. Otherwise,
 if the edge connects $u$ to a negative clause, we place a negative literal
 gadget in cell $(i, r)$. Since every edge of the rectilinear representation
 lies in an odd column, we can connect two literals of the same variable by
 either a copy or inverter gadget depending on whether both literals are
 positive or negative, or one is positive and the other negative.

 We substitute an edge $e$ of the rectilinear representation that connects a
 variable to a positive clause as follows. Let $i$ be the column of $e$. For
 every odd row $r$ that is covered by $e$ we place a positive literal gadget in
 cell $(i, r)$.  The copy gadget can be rotated in order to connect a literal
 gadget in cell $(i, r)$ to a literal gadget in a cell $(i, r+2)$.
 
 Let $R_c$ be the clause rectangle in the modified rectilinear representation of
 a positive clause $c$.  Let $i < j < k$ be the columns of the edges in the
 rectilinear representation that connect $R_c$ to the variables $x, y, z$.  We
 place the clause gadget in the clause rectangle such that literal gadget $Z$,
 refer to Fig.~\ref{fig:rectilinear_rep}, lies in column~$k$. Let $i_c$ denote
 the lowest row of $c$. Then the literal gadget of $X$ and $Y$ lie in column
 $k-4$ with $X$ in row $i_c + 4$ and $Y$ in column $i_c +2$. The clause block
 $K_c$ lies in cell $(i_c + 3, k+1)$.  We place a positive literal gadget $X'$
 in cell $(i, i_c + 4)$ and connect it with an alternating sequence of literal
 and copy gadgets to $X$.  Let $X_e$ be the literal gadget in cell $(i, i_c-2)$.
 Since between two variables there are at least three empty columns we can
 connect $X'$ to $X_e$ with an alternating sequence of literal and copy gadgets.
 Analogously, we connect $Y$ to the edge in column $j$ by placing a positive
 literal gadget in cell $(j, i_c + 2)$.  A negative clause is obtained by
 vertically mirroring the construction of a positive clause.

 We now argue that the embedding of the graph $\Cluster$ is planar and that the
 pairwise intersections of squares in the arrangement $\DiskMap$ are empty.
 Observe that, with the exception of the clause blocks $K_c$ every gadget is
 entirely embedded in the modified rectilinear representation. The column of
 $K_c$ is even, and therefore it cannot intersect with an edge of the
 rectilinear representation.  Recall that the rectilinear representation is
 planar and all gadget are placed in disjoint cells. Therefore, the pairwise
 intersection of squares in $\DiskMap$ is empty. Moreover, each literal gadget
 is planar embedded in a single cell, each clause is embedded in a rectangle that
 covers five rows and six columns, and finally each copy and inverter gadget
 together with its two literal gadget is embedded in either a single row and 3
 columns or in 3 rows and a single column. Thus, since the modified rectilinear
 representation is planar and the pairwise intersections of squares in $\DiskMap$
 are empty, the graph $\Cluster$ has a planar embedding.  Finally, the maximal
 vertex degree of the literal gadget is three, the maximal degree a clause
 gadget is four.  Connecting two literal gadgets by copy or inverter gadget
 increases the maximum vertex degree of $\Cluster$ to five. Further, the
 obstacle number of the literal gadget and clause gadget is one and the
 obstacle number of the copy and inverter gadget is two.
 
 It is left to show that the layout can be computed in polynomial time.  As
 already argued the modified rectilinear representation $\Pi$ of the monotone planar
 $3$-SAT instance can be computed polynomial time.  Moreover, the height and
 width of $\Pi$ is linear in $|C|$.  Thus, we inserted a number of gadgets
 linear in $|C|$. Further, the coordinates of each gadget are independent of the
 instance $(U, C)$, thus overall the representation of the final 
 arrangement $\DiskMap$ is polynomial in $|U|$ and $|C|$. Placing a single
 gadget requires polynomial time, thus overall the clustered graph $\Cluster$ and
 the arrangement $\DiskMap$ of squares is can be computed in polynomial time.
 
 \paragraph{Correctness}
 Assume that $(U, C)$ is satisfiable. Depending on whether a variable $u$ is
 true or false, we place all cycle vertices on a positive placement of a
 positive literal gadget and on the negative placement of negative literal
 gadget of the variable. Correspondingly, if $u$ is false, we place the vertices
 on the negative and positive placements, respectively.  By
 Property~\ref{lemma:valid_literal_gadget}, the placement induces a
 $\DiskMap$-framed drawing of all literal gadgets and
 Property~\ref{lemma:valid_copy_inverter_gadget} ensures the copy and inverter
 gadgets have a $\DiskMap$-framed drawing.  Since at least one variable of each
 clause is true, there is a $\DiskMap$-framed drawing of each clause gadget by
 Lemma~\ref{lemma::valid_clause_gadget}. 
 
 Now consider the clustered graph
 $\Cluster$ has a $\DiskMap$-framed drawing. Let $X$ and $Y$ be
 two positive literal gadgets or two negative literal gadgets connected with a
 copy gadget. By Lemma~\ref{lemma:consistent_copy_gadget}, a drawing of $X$ is
 positive if and only if the drawing of $Y$ is positive.
 Property~\ref{lemma:consistent_inverter_gadget} ensures that the drawing of a
 positive literal gadget $X$ is positive if and only if the drawing of the
 negative literal gadget $Y$ is negative, in case that both are joined
 with an inverter gadget.  Further, Lemma~\ref{lemma:no_infeasible_drawing}
 states that each cycle vertex lies either in a positive or negative region.
 Thus, the truth value of a variable $u$ can be consistently determined by any
 drawing of a positive or negative literal gadget of $u$.  By
 Lemma~\ref{lemma::valid_clause_gadget}, the clause gadget has no
 $\DiskMap$-framed drawing of the clause gadget such that all literal gadgets
 have a negative drawing. Thus, the truth assignment indeed satisfies $C$.
\end{proof}

Observe that we placed blue dotted circles in the figures of the gadgets. The
squares can be replaced by these circles without changing the essential
properties of the gadgets. More precisely, the disks can be replaced such that
they contain the positive and negative placements of the corresponding squares.
In case of obstacles, $O$ is tangent to some line $l$ in a point $p$, thus, $O$
is replaced by a disk that is tangent to $l$ in $p$. In the clause gadget we
placed the obstacles $O^2_A$ with respect to two intersecting lines $l_1, l_2$.
In this case we place a disk with small radius such that it is tangent to $l_1$
and $l_2$.  Therefore, the proof for squares can be adapted to proof that the
\textsc{\clusterDrawing} problem with disks is $\cNP$-hard.

\begin{corollary}
	The problem \textsc{\clusterDrawing} with disks is $\cNP$-hard, even when the
	clustered graph is restricted to vertex degree 5.
\end{corollary}

\fi
\end{document}